\documentclass[11pt]{article}
\usepackage{amsmath,amsfonts,epsf}
\usepackage{amssymb}
\usepackage{graphicx}
\usepackage{grffile}
\input epsf
\usepackage[nosort]{cite}
\usepackage{subcaption}
\usepackage{breqn}
\usepackage{epstopdf}
\usepackage{enumerate}
\usepackage{tensor}
\usepackage{mathrsfs}
\usepackage{amsmath}
\usepackage{amsthm}
\usepackage[english]{babel}
\usepackage{braket}
\usepackage{dcolumn}
\usepackage{bm}
\usepackage{hyperref}
\hypersetup{hidelinks}
\usepackage{slashed}

\usepackage{graphbox}

\usepackage{comment}

\makeatletter\@addtoreset{equation}{section}\makeatother

\theoremstyle{definition}
\newtheorem{definition}{Definition}[section]

\renewcommand{\title}[1]{\vbox{\center\LARGE{#1}}\vspace{5mm}}
\renewcommand{\author}[1]{\vbox{\center#1}\vspace{5mm}}
\newcommand{\address}[1]{\vbox{\center\em#1}}

\newcommand{\Tr}{\text{tr~}}

\begin{document}

\newtheorem{theorem}{Theorem}
\begin{titlepage}
\begin{center}
\vskip 1cm

\title{Qutrit and Ququint Magic States}

\author{Akalank Jain and Shiroman Prakash}

\address{Department of Physics and Computer Science, Dayalbagh Educational Institute, Dayalbagh, Agra, India 282005}

\end{center}

\begin{abstract}
Non-stabilizer eigenstates of Clifford operators are natural candidates for endpoints of magic state distillation routines. We provide an explicit bestiary of all inequivalent non-stabilizer Clifford eigenstates for qutrits and ququints. For qutrits, there are four non-degenerate eigenstates, and two families of degenerate eigenstates. For ququints, there are eight non-degenerate eigenstates, and three families of degenerate eigenstates. Of these states, a simultaneous eigenvector of all Clifford symplectic rotations known as the qutrit \textit{strange state} is distinguished as both the most magic qutrit state and the most symmetric qutrit state. We show that no analogue of the qutrit strange state (i.e., no simultaneous eigenvector of all symplectic rotations) exists for qudits of any odd prime dimension $d>3$.
\end{abstract}

\vfill

\end{titlepage}

\eject \tableofcontents

\clearpage
\section{Introduction}                       
Magic state distillation, first presented in \cite{MSD, knill}, is an approach to fault tolerant quantum computing \cite{natureReview} that relies on ancilla qudits in non-stabilizer states known as \textit{magic states} to promote the Clifford group to universal quantum computation. Magic state distillation for qubits is a rich and well-studied subject (see, e.g., \cite{MSD, knill, Reichardt2005, reichardt2009quantum, Haah1, low-overhead, CampbellHoward1, CampbellHoward2, HowardPredistilled, codyMultiLevel, catalysis, rall2017signed}). However, magic state distillation for qudits of other dimensions is less well-studied, but has attracted some interest \cite{ACB, CampbellAnwarBrowne, campbell2014enhanced, Howard} in the past few years. These works focus particularly on qudits of odd-prime dimension. 

There are several reasons to study magic state distillation with qudits of odd-prime dimension. There may be practical advantages to using qudits for magic state distillation instead of qubits as shown in e.g., \cite{campbell2014enhanced}. There are also important theoretical motivations to study qudit magic state distillation. In particular, qudit magic state distillation has been used to identify contextuality as a necessary and possibly-sufficient resource for universal quantum computation in \cite{nature, Delfosse_2017}. Moreover, as shown in \cite{noiseQudit,Veitch_2012, MariEisert, Veitch_2014,  wang2018efficiently}, the existence of discrete Wigner functions offers  new possibilities for defining a resource theory of non-stabilizer states not available for qubits (or, more generally, quantum systems of even dimension). Given these motivations, we feel it is worthwhile to explore the state space of qudits of the smallest odd dimensions, i.e., qutrits and ququints. These dimensions also happen to be prime, which simplifies our analysis. We also remark that a variety of experimental realizations of qutrits do exist (e.g., \cite{mair2001entanglement,PhysRevLett.105.223601,PhysRevA.67.062313}.)

Both of the qubit magic states identified in \cite{MSD}, $\ket{H}$ and $\ket{T}$ are eigenstates of Clifford operators. Their symmetry properties under Clifford transformations can easily be visualized using the stabilizer octahedron inscribed within Bloch sphere. For example, $\ket{H}$ states lie directly  above edges of the stabilizer octahedron, and hence there are 12 in total. $\ket{T}$ states lie directly above faces of the stabilizer octahedron, and hence there are 8 in total. $\ket{T}$ states are farther from the stabilizer octahedron than $\ket{H}$ states, so they may, in principle, be distilled with higher threshold to noise. In contrast, qudit state space is substantially more abstract -- see, e.g., \cite{Kimura, Mendas, Ingemar, Bengtsson2, Appleby5, DPS2, goyal2016geometry} for detailed discussion of qudit and qutrit state space in particular. Many basic features about qudit  states therefore remain relatively opaque. To overcome this difficulty, we perform some elementary but hopefully useful calculations with a view towards facilitating future work on qudit magic state distillation.

In particular, to identify states that may be candidates for endpoints of magic state distillation routines, we explicitly enumerate the eigenstates of single-qutrit and single-ququint Clifford operators, and study their symmetries under Clifford transformations.\footnote{It is not necessary in principle, for magic states to be eigenstates of Clifford operators. For example, the qubit state $\ket{\pi/3}=\frac{1}{\sqrt{2}} \left(\ket{0}+e^{i\pi/3} \ket{1} \right)$ is the endpoint of a distillation routine given in \cite{Howard}, and can be used for state injection, though it is not an eigenstate of a Clifford operator. However it appears easier to construct distillation schemes for Clifford eigenstates.} We consider candidate magic states to be equivalent if they related to each other by a single-qudit Clifford unitary.

For qutrits, we find that there are four inequivalent non-degenerate non-stabilizer Clifford eigenstates, and two inequivalent one-complex-parameter families of degenerate qutrit Clifford eigenstates. For ququints, we find that there are eight inequivalent non-degenerate non-stabilizer Clifford eigenstates, three inequivalent one-complex-parameter families of degenerate Clifford eigenstates. and one two-complex-parameter family of degenerate Clifford eigenstates. For qutrits, most of these states have been previously identified in the papers \cite{Veitch_2014, noiseQudit, ACB, CampbellAnwarBrowne, Howard}, in the context of qutrit magic state distillation; though an one of the degenerate families of states has not been previously discussed in the literature. For ququints, most of these states appear unstudied, but some discussion of the equatorial magic state (the state $\ket{XV_{\hat{S}}}$ defined below) appears in \cite{campbell2014enhanced, HowardVala, maxnonlocality}, and also \cite{noiseQudit}. These states are depicted in Figure \ref{fig:venn} and Figure \ref{ququintVenn}, which illustrates their symmetry properties under Clifford transformations.

The \textit{mana} of a state, introduced in \cite{Veitch_2014},  is a calculable measure of its usefulness as a magic state. The mana is only defined for qudits of odd dimension. We are able to provide an alternative characterization of the four non-degenerate qutrit magic states as local maxima of the mana.  The qutrit states with maximal mana are the strange state, $\ket{S}$, and the Norell state $\ket{N}$, defined in section \ref{qutrit-states} below. The most magic ququint state is $\ket{B,-e^{2\pi i /3}}$, defined in section \ref{ququint-states}. We numerically verified that this state maximizes the mana over all pure ququint states.

The size of the orbit of each state under the Clifford group is a measure of how ``symmetric'' a state is, under Clifford transformations, with smaller orbits corresponding to more symmetric states. We compute this quantity for each of the non-degenerate eigenstates, and find that the strange state is the most symmetric qutrit Clifford eigenstate, while the state $\ket{B,-1}$ defined in section \ref{ququint-states} is the most symmetric ququint Clifford eigenstate.


In section \ref{twirl-section}, we  show how to use these results to construct twirling schemes that can be used prior magic state distillation for each of the candidate magic states. Some qutrit magic state distillation routines were presented in \cite{ACB, CampbellAnwarBrowne, Howard}. Distillation routines with the state $\ket{S}$ were not known until the recent work \cite{GolaySP}.

The qutrit strange state is distinguished as not only the most magic state but also the most symmetric qutrit Clifford eigenstate. The size of its orbit under the Clifford group is $9$, because it is a simultaneous eigenvector of the set of symplectic rotations. In section \ref{qutrit-strange-state}, we show that the existence of such a state is unique to $d=3$ -- there is no simultaneous eigenvector of symplectic rotations for any higher odd prime. The qutrit strange state, is therefore, particularly special amongst all qudit Clifford eigenstates.

Related work, in the context of SIC-POVMs, appears in \cite{Appleby1, Appleby2, Appleby3, Appleby4,subgroup, Appleby5}. Alternative approaches to magic states as eigenstates of permutation operators appear in \cite{planat2017magic, planat2017magic2, sym10120773}.

\section{Preliminaries}
\label{sec-2}
In this section, we very briefly review some basic facts associated with fault-tolerant quantum computing using qudits of odd-prime dimension, and establish the conventions and notations used throughout the paper. The facts that we summarize here have been widely used in the literature: see, e.g., \cite{Gross, Veitch_2012, Veitch_2014, HowardVala, nature, GolaySP} for similar summaries.

 Henceforth, we will use the word qudit to refer to a quantum system of odd prime dimension $p$. Let $\omega_p=e^{2\pi i /p}$ and $i,j,k,n \in \mathbb Z_p$. In this paper, we will focus exclusively on single-qudit operators and states. 

\begin{definition}The higher-dimensional generalization of the single-qubit Pauli group, denoted as $\mathcal P_p$ is, known as the \textbf{single-qudit Heisenberg-Weyl displacement group}, is the group is generated by the operators
\begin{equation}
X \ket{n} = \ket{n+1}, ~ Z \ket{n} = \omega_p^n \ket{n},
\end{equation}
which satisfy $ZX=\omega_p XZ$. 
\end{definition}
There are $p^2$ linearly-independent  Heisenberg-Weyl displacement operators, which can be parameterized using a symplectic vector $\chi=(u,v)$, where  $u,~v \in \mathbb Z_p$. 
These are conventionally defined as follows:
\begin{equation}
 D_{\chi} = \omega_p^{uv2^{-1}}X^uZ^v. \label{define-HW}
\end{equation}

 Each Heisenberg-Weyl displacement operator (other than the identity) has $p$ orthogonal eigenstates, with eigenvalues $\omega_p^k$:
\begin{equation}
    D_{\chi} \ket{\chi; k} = \omega_p^k \ket{\chi;k}
\end{equation} 
Eigenstates of Heisenberg-Weyl displacement operators are known as \textit{stabilizer states}. We denote the set of stabilizer states as $\mathcal S$. Note that both $D_{\chi}$ and $D_{\chi}^n=D_{n\chi}$ have the same eigenvectors. To count the total number of stabilizer states in $p$ dimensions, we observe that there are $p+1$ linearly independent symplectic vectors: $(1,0)$, $(0,1)$, $(0,2)$, $\ldots, (0,p-1)$. Therefore, there are $p(p+1)$ distinct single-qudit stabilizer states, as reviewed in, e.g. \cite{Gross}.

\subsection{The Single-Qudit Clifford Group and $SL(2,\mathbb Z_p)$}

\begin{definition}
The \textbf{single-qudit Clifford group}, $\mathcal C$, is defined as the set of all unitary operators that map qudit Pauli operators to qudit Pauli operators:
\begin{equation}
\mathcal C = \{C \in U(p)| C\mathcal P_p C^\dagger = \mathcal P_p \}.
\end{equation}
\end{definition}
Note that, as a consequence of this definition, Clifford transformations map stabilizer states onto stabilizer states. While it is not difficult to generalize this definition to multiple qudits, in this paper we will focus our attention exclusively on the single-qudit Clifford group. Since we are concerned with the action of the Clifford group on qudit density matrices, we will often ignore unphysical overall phases, in the discussion that follows. 

A crucial property of single-qudit Clifford operators that we will use extensively in this paper is that single-qudit Clifford operators act on Heisenberg-Weyl displacement operators as $SL(2,\mathbb Z_p)$ transformations. As shown in \cite{Appleby1}, one way of stating this fact is that the single-qudit Clifford group is isomorphic to the semi-direct product $\mathbb Z_p^2 \rtimes SL(2,\mathbb Z_p)$.  We summarize this result as the following theorem, which we state without proof, 
\begin{theorem} \label{SLp}
Up to an overall phase, any single qudit Clifford operator $C$ can be expressed as
\begin{equation}
    C = D_{\chi}V_{\hat{F}}
\end{equation}
where $\chi \in \mathbb Z_p^2$, $\hat{F}  = \begin{pmatrix} a & b \\ c & d \end{pmatrix} \in SL(2,\mathbb Z_p)$, and $V_{\hat{F}}$ is given by:
\begin{equation}
V_{\hat{F}}=
\begin{cases}
\frac{1}{\sqrt{p}} \displaystyle \sum_{j,k=0}^{p-1} \omega_p^{2^{-1} b^{-1}(ak^2-2jk+dj^2)}\ket{j}\bra{k} & b\neq 0 \\
\displaystyle \sum_{k=0}^{p-1}\omega_p^{2^{-1} ack^2} \ket{ak}\bra{ k} & b= 0 \\
\end{cases}.
\end{equation}
Note that $a$, $b$, $c$ and $d \in \mathbb Z_p$, and satisfy $ad-bc=1$. $2^{-1}$ denotes the inverse of $2$ in the finite field $\mathbb Z_p$.

This is an isomorphism, in that, the following relation is obeyed (up to an overall phase),
\begin{equation}
D_{\chi_1}V_{\hat{F_1}} D_{\chi_2}V_{\hat{F_2}} = D_{\chi_1+\hat{F}_1 \chi_1}V_{\hat{F_1}\hat{F_2}}.
\end{equation}
\end{theorem}
It is usually much more convenient to specify Clifford operators as elements of $\mathbb Z_p^2 \rtimes SL(2,\mathbb Z_p)$, rather than as $p\times p$ matrices in $U(p)$.

\textit{Notation:} We will exclusively use capital letters with overhats, i.e., $\hat{F}$,  to denote elements of $SL(2,\mathbb Z_p)$, i.e., $2\times 2$ matrices. Capital letters without overhats denote elements of $U(p)$, i.e. $p\times p$ matrices. 

The group $SL(2,\mathbb Z_p)$ can be generated by two elements
\begin{equation}
\hat{H} = \begin{pmatrix}
0 & 1 \\ -1 & 0
\end{pmatrix}, \text{ and }
\hat{S} = \begin{pmatrix}
1 & 0 \\
1 & 1 \end{pmatrix}.
\end{equation}

The complete single-qudit Clifford group is generated by the operators $H=V_{\hat{H}}$ and $S=D_{(2^{-1},0)}V_{\hat{S}}$, which are explicitly given by:
\begin{equation}
H\ket{k}=\frac{1}{\sqrt{p}}  \sum_j \omega_p^{jk} \ket{j},  ~V_{\hat{S}} \ket{k} = \omega_p^{2^{-1} k^2}\ket{k}.
\end{equation}

\textit{Notation:} As is common in the mathematics literature, we will use the notation $\langle a, b, c, \ldots \rangle$ to denote the group generated by $a, b, c, \ldots$, using a multiplication rule that is clear from context. We can therefore write $SL(2,\mathbb Z_p)=\langle \hat{H}, \hat{S} \rangle$, or $\mathcal C = \langle H, S \rangle$. We will also use this notation to specify subgroups of either $SL(2,\mathbb Z_p)$ or the single-qudit Clifford group. For example, $\langle H \rangle$ denotes the subgroup of $SL(2,\mathbb Z_p)$ consisting of the four elements: $\{ 1, \hat{H}, \hat{H}^2, \hat{H}^3 \}$ isomorphic to $\mathbb Z_4$.

\subsection{Discrete Phase Space}
Wigner first observed in \cite{wigner1932} that quantum states can be represented as quasi-probability distributions over phase space, known as Wigner functions. Discrete Wigner functions, defined in  \cite{Wootters1987, PhysRevA.70.062101}, are the analog of Wigner's construction for finite-dimensional quantum systems of prime dimension for which phase space is $\mathbb Z_p\otimes \mathbb Z_p$. In particular, the following construction of a discrete Wigner function, discussed extensively in \cite{Appleby2,Gross, cormick2006classicality}, is particularly useful, and has been employed in \cite{Veitch_2012,nature}.

Note that qubit Pauli operators are not only unitary but also Hermitian, and can therefore be used as a basis for expressing single-qubit density matrices. However, the single-qudit Heisenberg-Weyl operators defined in equation \ref{define-HW} (for odd prime dimensions $p$) are unitary but not Hermitian. To construct a manifestly Hermitian basis for single-qudit density operators, we define the \textit{discrete phase point operators}, as follows:
\begin{eqnarray}
A_{(0,0)}&=&\frac{1}{p}\sum_{u=0}^{p-1}\sum_{v=0}^{p-1} D_{(u,v)} \\ A_{\chi}&=&D_{\chi}A_{(0,0)} D_{\chi}^\dagger.
\end{eqnarray}

The discrete Wigner function for a state ${\rho}$ is defined in terms of the phase point operators as:
\begin{equation}
W_{\chi}(\rho)=\frac{1}{p} \text{tr }(\rho A_{\chi}).
\end{equation}
The entries of the discrete Wigner function are $p^2$ real numbers that completely characterize the quantum state $\rho$. For normalized density matrices, only $p^2-1$ of these numbers are independent, and the following relation is satisfied:
\begin{equation}
    \sum_{\chi}W_{\chi}=1.
\end{equation}
When the Wigner function is non-negative, it can be interpreted as a probability distribution on phase space \cite{Wootters1987, Gross}. The discrete Wigner function is a convenient way to visualize higher-dimensional qudit states. 

Clifford operators act as translations and symplectic rotations on the phase point operators:
\begin{equation}
    (D_\chi V_{\hat F}) A_\psi (D_\chi V_{\hat F})^\dagger = A_{\hat{F}\psi+\chi}.
\end{equation}
We refer to the subgroup of the Clifford group consisting of operators of the form $V_{\hat{F}}$ as the set of \textit{symplectic rotations}, and the set of Heisenberg-Weyl displacement operators as $\textit{symplectic translations}$. These operations can therefore be efficiently simulated acting on states with non-negative discrete Wigner function, as explained in \cite{Veitch_2012}.

The set of states with non-negative discrete Wigner function is a convex set known as the Wigner polytope. Perhaps surprisingly, as observed in \cite{Veitch_2012}, the Wigner polytope contains the stabilizer polytope, which is the set of mixtures of stabilizer states, as a proper subset. Because Clifford operations on states within the Wigner polytope can be efficiently simulated, the ability to prepare ancillas outside of the Wigner polytope is a necessary condition for magic state distillation. In \cite{nature}, it was also shown that negativity of the Wigner function is equivalent to contextuality with respect to stabilizer measurements. It was also recently shown in \cite{bound} that no finite magic state distillation routine can distill states tight to the boundary of any facet of the Wigner polytope \cite{qudit-bound-states}, generalizing the well-known result for qubits presented in \cite{structure}. (See also \cite{PhysRevLett.115.070501, Delfosse_2017, cormick2006classicality, Howard_2013} for related work.)

\subsection{Measures of Magic}
In the magic state model of quantum computation, ancilla qudits in non-stabilizer states serve as a resource for achieving universal quantum computation. This idea is made precise in \cite{Veitch_2014}, who formulate a resource theory for non-stabilizer states, that is inspired by the analogous resource theory of entanglement.
In \cite{Veitch_2014} a \textit{magic state} is defined to be any pure quantum state that is not a stabilizer state. The amount of magic a state possesses can be quantified via the \textit{regularized entropy of magic}. Let us denote the relative entropy between two states $\rho$ and $\sigma$ as
\begin{equation}
    D(\rho;\sigma) = \Tr \left[ \rho (\log \rho-\log \sigma) \right].
\end{equation}

\begin{definition}
The \textbf{regularized entropy of magic}, for a state $\rho$, is defined to be $1/n$ times the minimum relative entropy between $\rho^{\otimes{n}}$ and any $n$-qudit stabilizer state, $\sigma$, in the limit $n \rightarrow \infty$:
\begin{equation}
    \mathcal R_M(\rho) = \lim_{n\rightarrow \infty}~ \frac{1}{n} \min_{\sigma \in \mathcal S}D(\rho^{\otimes n};\sigma).
\end{equation}
\end{definition}
The regularized entropy of magic has several attractive properties that justify treating it as a resource, as explained in \cite{Veitch_2014}.

Unfortunately, the regularized entropy of magic is not possible to compute. Thanks to the existence of the discrete Wigner function in odd dimensions, a powerful computable alternative to regularized entropy of magic exist, which is the \textit{mana}, defined in \cite{Veitch_2014}. 


\begin{definition}The \textbf{mana}, $\mathcal M(\rho)$ of a state $\rho$ is defined as, \begin{equation}
\mathcal M(\rho)= \log \left( \sum_{p,q} |W_{(p,q)}(\rho)| \right). \label{calculate-mana}
\end{equation}
\end{definition}
To better understand the physical interpretation of mana, it is convenient to define the \textit{sum negativity}, $\text{sn}(\rho)$, as the absolute value of the sum of negative entries in the discrete Wigner function of $\rho$:
\begin{equation}
    \text{sn}(\rho)=\sum_{\chi:~W_\chi(\rho)<0} |W_\chi(\rho)|
\end{equation}
For states which are normalized, the mana can also be expressed in terms of $\text{sn}(\rho)$ as:
\begin{equation}
\mathcal M(\rho) = \log (2 \text{sn}(\rho)+1). \label{sn-mana}
\end{equation}

Because negative entries act as an obstacle to classical simulation, it is satisfying that the sum-negativity can also be used to define a resource for quantum computation via equation \ref{sn-mana}. In fact, it was shown in \cite{Veitch_2014} that  mana is the only meaningful resource that can be defined from negativity of the Wigner function. In particular, the most-negative entry of the discrete Wigner function of a magic state is not a meaningful resource, although it does determine the best theoretical threshold of a magic state distillation scheme that distills the given state, as shown in \cite{noiseQudit}.

Another computable measure of magic, known as  \textit{thauma} was defined recently in \cite{wang2018efficiently, wang_2019}. We remark that it would be interesting to combine the above measures of magic, with the graph theoretic formalism of \cite{graphTheory1,graphTheory2}. Our discussion focused on qudits of odd prime dimension, some related work for qubits appears in \cite{QubitMagic1,QubitMagic2}.

\section{Eigenstates of Qutrit Clifford Operators}
\label{qutrit-states}

Any pure quantum state that is not a stabilizer state could be considered a magic state, as per the definition of magic, given in \cite{Veitch_2014}. Indeed, one generically expects that it is possible to use any such state to implement non-Clifford gates via state injection. However, in qubit magic state distillation, the authors of \cite{MSD} reserve the term magic state for two particular single-qubit non-stabilizer states: $\ket{H}$ and $\ket{T}$. These two states are considered magic because, not only can they be used for state-injection, but they can also be  \textit{distilled} via a magic state distillation protocol. A magic state distillation protocol takes many low-fidelity copies of the magic state as input, and produces a single higher-fidelity magic state, using only Clifford unitaries and stabilizer measurements. Here, we ask,  what are the higher-dimensional analogues of Bravyi and Kitaev's $\ket{T}$ and $\ket{H}$ states?

An important property of $\ket{T}$ and $\ket{H}$ states is that they are eigenvectors of single-qubit Clifford operators. In particular $\ket{H}$ is an eigenstate of the single-qubit Hadamard gate, and $\ket{T}$ is an eigenstate of the single-qubit Clifford operator $\frac{e^{i\pi/4}}{\sqrt{2}}\begin{pmatrix} 1 & 1 \\ i & -i \end{pmatrix}$. This is manifestly apparent from their positions in the Bloch sphere, directly above the faces (or edge) of the stabilizer octahedron. It is clear from the analysis in \cite{MSD}, and elsewhere, that this property plays an important role in facilitating the construction and analysis of distillation protocols. 

Motivated by this observation, we give the following definition.
\begin{definition}
Define a single-qudit state $\ket{\psi}$ to be a \textbf{qudit magic state}, if $\ket{\psi}$ is an eigenstate of a single-qudit Clifford unitary, and $\ket{\psi}$ is not a stabilizer state. 
\end{definition}Strictly speaking, some of these states should be considered \textit{candidate} magic states, because distillation protocols for these states have not yet been constructed.  

Using the above definition for qubits, we find that there are a total of 20 magic states. However, by applying single-qubit Clifford transformation, 12 of these states (those lying above the edges of the stabilizer octahedron) can be mapped to $\ket{H}$ and the remaining 8 states (those lying directly above the faces of the stabilizer octahedron) can be mapped to $\ket{T}$.  

By analogy, we define the following natural notion of equivalence.
\begin{definition} We define two candidate magic states to be \textbf{Clifford-equivalent}, or simply \textit{equivalent}, if they are related by a \textit{single-qudit} Clifford operation: 
\begin{equation}
    \ket{m_1} \sim \ket{m_2},~\text{if there exists a }C \in \mathcal C \text{, such that } \ket{m_1}=C\ket{m_2}.
\end{equation}
If this is not the case, we say the two states are \textbf{Clifford-inequivalent}.
\end{definition}

Let us emphasize that this our definition of Clifford-equivalence is deliberately restricted to equivalence via single-qudit Clifford unitaries for operational simplicity. Any magic state distillation protocol or state injection scheme for a magic state $\ket{m}$ can be trivially adapted to also work for $C\ket{m}$. Therefore, if any two magic states are related to each other by a single-qudit Clifford transformation, then they should clearly be regarded as equivalent under any reasonable definition of equivalence.

However, if one allows for multi-qudit operations, then more general notions of equivalence may be possible. Given a sufficient number of copies, $N_m$, of a magic state $\ket{m}$, it may be possible to obtain $N_\psi$ copies of another quantum state $\ket{\psi}$, using multi-qudit Clifford unitaries and stabilizer measurements (via, e.g., state-injection). An interesting question is then, what is the best possible asymptotic rate of conversion, $R(\ket{\psi},\ket{m})=N_\psi/N_m$, that is attainable, and for which states $\ket{\psi}$ is this rate equal to one? A necessary  condition, for $R(\ket{\psi},\ket{m})=1$ is that $\ket{\psi}$ and $\ket{m}$ have an identical relative entropy of magic.\cite{Veitch_2014} This is guaranteed to be the case if $\ket{\psi}=C\ket{m}$ for some Clifford unitary $C$, but it is in principle possible that there are other single-qudit states $\ket{\phi}$ not related to $\ket{\psi}$ by a single-qudit Clifford unitary that also possess the same relative entropy of magic. We do not pursue this more general notion of equivalence here, but we note that some interesting bounds on inter-conversion of qutrit states appear in \cite{wang2018efficiently}. (It turns out that all the non-degenerate magic states identified in this paper have unequal mana or thauma, hence they must be inequivalent even in this more general sense.)

Below we will identify all Clifford-inequivalent magic states for qutrits and ququints, and compute some of their basic properties.  

Before we proceed, let us observe that we could have instead chosen to define  a magic state as any single-qudit state which is a local maximum of the relative entropy of magic. This would perhaps be a more meaningful definition, but unfortunately, the relative entropy of magic for an arbitrary single-qutrit state is not feasible to compute, as far as we know. However, in the appendix, we analytically determine all qutrit states which locally maximize the mana; and we find that these states coincide with the non-degenerate single-qutrit Clifford eigenstates identified in this section. 

\subsection{Conjugacy Classes of the Single-Qutrit Clifford Group}

The process of identifying all Clifford-inequivalent qudit magic states is facilitated greatly by computing the conjugacy classes of the single-qudit Clifford group.
Recall that two elements, $g_1$ and $g_2$, of a group $G$ are said to be in the same \textit{conjugacy class}, denoted as $[g_1]$ or $[g_2]$, if there exists an $h \in G$ such that $hg_1h^{-1}=g_2$

Suppose two Clifford operators $C_1$ and $C_2$ are in the same conjugacy class $[C_1]$ of the Clifford group; then it is easy to see that the eigenvectors of $C_1$ are related to the eigenvectors of $C_2$ by a Clifford transformation. Hence, to enumerate all Clifford-inequivalent magic states, we only need to compute eigenvectors of one representative of each conjugacy class of the Clifford group. 

Note that, if $C \in \mathcal C$, the conjugacy class $[C]$ and $[C^{-1}]$ may be different, but both $C$ and $C^{-1}$ have identical eigenspaces.  To simplify the subsequent analysis, we define an additional equivalence for elements of the Clifford group. 

\begin{definition}We say that two different elements of the single-qudit Clifford group, $C_1$ and $C_2$ are \textbf{eigenspace-equivalent} if there exists a single-qudit Clifford unitary $C$ that acts as a bijection from the eigenspace of of $C_1$ to the eigenspace of $C_2$.  We define the equivalence classes of Clifford operators with respect to this equivalence this equivalence relation as set of \textbf{reduced conjugacy classes}.\end{definition} 

We denote the reduced conjugacy classes as $[[C]]$. Note that if $C_1$ and $C_2$ belong to the same conjugacy class of the single-qudit Clifford group, they automatically belong to the same reduced conjugacy class.

To enumerate all inequivalent qudit magic states, we must enumerate the eigenvectors of one representative of each reduced conjugacy class of the single-qudit Clifford group. 
In a slight abuse of terminology we may use the phrase ``eigenvectors of a conjugacy class'' to mean the eigenvectors of a representative operator of that reduced conjugacy class.

\subsubsection*{Conjugacy Classes of $SL(2,\mathbb Z_3)$}
As reviewed in section \ref{sec-2}, single qutrit Clifford unitaries are (up to an overall phase) in one-to-one correspondence with $\mathbb Z_3^2 \rtimes SL(2,\mathbb Z_3)$. The easiest way to explicitly compute the conjugacy classes of the single-qutrit Clifford group is to first compute the conjugacy classes of $SL(2, \mathbb Z_3)$.  The conjugacy classes of $SL(2,\mathbb Z_3)$ are well-known, (see e.g., \cite{Appleby2, Humphreys}), and can also be computed directly using a computer algebra system without difficulty. We find they are given by the following,
\begin{eqnarray}
\left[ \hat{\mathbb{I}} \right]  & = &
\left\{\left(
\begin{array}{cc}
 1 & 0 \\
 0 & 1 \\
\end{array}
\right)\right\}
\\
\left[ -\hat{\mathbb{I}} \right] & = &
\left\{\left(
\begin{array}{cc}
 -1 & 0 \\
 0 & -1 \\
\end{array}
\right)\right\} \\
\left[ \hat{H} \right] & = &
\left\{\left(
\begin{array}{cc}
 0 & 1 \\
 2 & 0 \\
\end{array}
\right),\left(
\begin{array}{cc}
 0 & 2 \\
 1 & 0 \\
\end{array}
\right),\left(
\begin{array}{cc}
 1 & 1 \\
 1 & 2 \\
\end{array}
\right),\left(
\begin{array}{cc}
 1 & 2 \\
 2 & 2 \\
\end{array}
\right),\left(
\begin{array}{cc}
 2 & 1 \\
 1 & 1 \\
\end{array}
\right),\left(
\begin{array}{cc}
 2 & 2 \\
 2 & 1 \\
\end{array}
\right)\right\} \\
\left[ \hat{S} \right]  & = &
\left\{\left(
\begin{array}{cc}
 1 & 0 \\
 1 & 1 \\
\end{array}
\right),\left(
\begin{array}{cc}
 1 & 2 \\
 0 & 1 \\
\end{array}
\right),\left(
\begin{array}{cc}
 2 & 2 \\
 1 & 0 \\
\end{array}
\right),\left(
\begin{array}{cc}
 0 & 2 \\
 1 & 2 \\
\end{array}
\right) \right\} \\
\left[ \hat{S}^2 \right]  & = &
\left\{\left(
\begin{array}{cc}
 1 & 0 \\
 2 & 1 \\
\end{array}
\right),\left(
\begin{array}{cc}
 1 & 1 \\
 0 & 1 \\
\end{array}
\right),\left(
\begin{array}{cc}
 2 & 1 \\
 2 & 0 \\
\end{array}
\right),\left(
\begin{array}{cc}
 0 & 1 \\
 2 & 2 \\
\end{array}
\right)\right\} \\
\left[ \hat{N} \right] & = &
\left\{\left(
\begin{array}{cc}
 2 & 0 \\
 2 & 2 \\
\end{array}
\right),\left(
\begin{array}{cc}
 2 & 1 \\
 0 & 2 \\
\end{array}
\right),\left(
\begin{array}{cc}
 1 & 1 \\
 2 & 0 \\
\end{array}
\right), \left(
\begin{array}{cc}
 0 & 1 \\
 2 & 1 \\
\end{array}
\right)\right\}\\
\left[ \hat{N}^{-1} \right]  & = &
\left\{\left(
\begin{array}{cc}
 2 & 0 \\
 1 & 2 \\
\end{array}
\right),\left(
\begin{array}{cc}
 2 & 2 \\
 0 & 2 \\
\end{array}
\right),\left(
\begin{array}{cc}
 1 & 2 \\
 1 & 0 \\
\end{array}
\right),\left(
\begin{array}{cc}
 0 & 2 \\
 1 & 1 \\
\end{array}
\right)\right\}
\end{eqnarray}
Here, we have defined $\hat{N}=\hat{S}\hat{H}^2$.

Let us also list the subgroups of $SL(2,\mathbb Z_3)$ for future use. There are 14 proper subgroups of $SL(2,\mathbb Z_3)$, which can be divided into 6 conjugacy classes. Apart the trivial subgroup consisting of only the identity, these are:
\begin{enumerate}
\item One subgroup of size $8$ isomorphic to the quaternion group, generated by $\hat{H}$ and $\hat{H}'=\left(
\begin{array}{cc}
 2 & 2 \\
 2 & 1 \\
\end{array}
\right)$.
\item Four subgroups of size $6$ isomorphic to $\mathbb Z_6$, each generated by one of the elements of $[\hat{N}]$.
\item Three subgroups of size $4$ isomorphic to $\mathbb Z_4$, which the three groups generated by $\hat H$, $\hat H'$, and $\hat H''=\hat H \hat H'$.
\item Four subgroups of size $3$ isomorphic to $\mathbb Z_3$, each generated by one of the elements of $[\hat{S}]$.
\item One subgroup of size $2$ isomorphic to $\mathbb Z_2$, generated by $\hat{2}$.
\end{enumerate}

\subsubsection*{Conjugacy Classes of $\mathbb Z_3^2 \rtimes SL(2,\mathbb Z_3)$}
We now turn our attention to conjugacy classes of the single qutrit Clifford group, which is isomorphic to $\mathbb Z_3^2 \rtimes SL(2,\mathbb Z_3)$. By direct computation, we find four of the conjugacy classes of $SL(2,\mathbb Z_3)$ directly translate into the following conjugacy classes of the Clifford group:
\begin{eqnarray}
[H] &=& \{ D_{\vec{\chi}} V_{\hat{F}} ~|~ \chi \in \mathbb Z_3^2,~\hat{F} \in [\hat{H}] \} \\
\left[ N \right] &=& \{ D_{\vec{\chi}} V_{\hat{F}} ~|~ \chi \in \mathbb Z_3^2,~\hat{F} \in [\hat{N}] \} \\
\left[ N^{-1} \right] &=& \{ D_{\vec{\chi}} V_{\hat{F}} ~|~ \chi \in \mathbb Z_3^2,~\hat{F} \in [\hat{N}^{-1}] \} \\
\left[ V_{-\hat{\mathbb{I}}} \right] &=& \{ D_{\vec{\chi}} V_{\hat{F}} ~|~ \chi \in \mathbb Z_3^2,~\hat{F} \in [-\hat{\mathbb{I}}] \}.
\end{eqnarray}
where we defined $N = V_{\hat{N}}$.

The conjugacy class of the identity in $SL(2,\mathbb Z_3)$ splits into two different conjugacy classes in $\mathbb Z_3^2 \rtimes SL(2,\mathbb Z_3)$:
\begin{eqnarray}
\left[ \mathbb I \right] &=& \{\mathbb  I\} \\
\left[ \text{Pauli}  \right] &=& \{ D_{\vec{\chi}} ~|~ \chi \in \mathbb Z_3^2, \chi \neq (0,0) \}.
\end{eqnarray}

The conjugacy classes $[S]$ and $[S^2]$, (both of which contain an element of the form $\begin{pmatrix} 1 & 0 \\ \gamma & 1 \end{pmatrix}$), also both split into two different conjugacy classes:
\begin{eqnarray}
\left[ V_{S^2} \right] & = & \left \{ D_{(0,v)} V_{\tiny \begin{pmatrix} 1 & 0 \\ 2 & 1 \end{pmatrix}}, ~D_{(u,0)} V_{\tiny \begin{pmatrix} 1 & 1 \\ 0 & 1 \end{pmatrix} },~
D_{(u,-u)} V_{\tiny \begin{pmatrix} 2 & 1 \\ 2 & 0 \end{pmatrix}},~D_{(u,u)}V_{\tiny \begin{pmatrix} 0 & 1 \\ 2 & 2 \end{pmatrix}} \right \} \\
\left[ XV_{S^2} \right] & = & \left \{ D_{(u^
*,v)} V_{\tiny \begin{pmatrix} 1 & 0 \\ 2 & 1 \end{pmatrix}}, ~D_{(u,v^*)} V_{\tiny \begin{pmatrix} 1 & 1 \\ 0 & 1 \end{pmatrix}},~
D_{(u,-u+v^*)} V_{\tiny \begin{pmatrix} 2 & 1 \\ 2 & 0 \end{pmatrix}},~D_{(u,u+v^*)}V_{\tiny \begin{pmatrix} 0 & 1 \\ 2 & 2 \end{pmatrix}} \right \} \\
\left[ V_{\hat{S}} \right] & = &\left\{ D_{(0,v)} V_{\tiny \begin{pmatrix} 1 & 0 \\ 1 & 1 \end{pmatrix}}, ~D_{(u,0)} V_{\tiny \begin{pmatrix} 1 & 2 \\ 0 & 1 \end{pmatrix}},~
D_{(u,-u)} V_{\tiny \begin{pmatrix} 0 & 2 \\ 1 & 2 \end{pmatrix}},~D_{(u,u)}V_{\tiny \begin{pmatrix} 2 & 2 \\ 1 & 0 \end{pmatrix}} \right\} \\
\left[ XV_{\hat{S}} \right] & = &\left\{ D_{(u^
*,v)} V_{\tiny \begin{pmatrix} 1 & 0 \\ 1 & 1 \end{pmatrix}}, ~D_{(u,v^*)} V_{\tiny \begin{pmatrix} 1 & 2 \\ 0 & 1 \end{pmatrix}},~
D_{(u,-u+v^*)} V_{\tiny \begin{pmatrix} 0 & 2 \\ 1 & 2 \end{pmatrix}},~D_{(u,u+v^*)}V_{\tiny \begin{pmatrix} 2 & 2 \\ 1 & 0 \end{pmatrix}} \right\}.
\end{eqnarray}
In the above expressions, $u$ and $v$ denote any element of $\mathbb Z_3$; and $u^*$ and $v^*$ denote any non-zero element of $\mathbb Z_3$.

In summary, we find that there are a total of $10$ conjugacy classes for the single-qutrit Clifford group.

\subsubsection*{Eigenvectors of Clifford Conjugacy Classes}

Let us now calculate the eigenvectors of a representative operator for each conjugacy class. 

In what follows, we denote the eigenvector of an operator $A$ with eigenvalue $a$ as $\ket{A,a}$. If the eigenspace of $A$ is degenerate, we write the state as $\ket{A,a;x,~y,\ldots}$, where $x,~y,\ldots$ parameterize the degenerate eigenspace.

The reduced conjugacy class $[\mathbb{I}]$ is trivial, and the eigenvectors of $[\text{Pauli}]$ are stabilizer states. The eigenvectors of the remaining conjugacy classes are as follows.

\subsubsection*{Eigenvectors of $[V_{-\hat{\mathbb{I}}}]$}
 The eigenvectors of $V_{-\hat{\mathbb{I}}}$ include the 2-dimensional family of degenerate eigenvectors:
\begin{equation}
\ket{V_{-\hat{\mathbb{I}}}, 1; a,b} = a\ket{0}+b(\ket{1}+\ket{2}),
\end{equation}  and the non-degenerate eigenvector:
\begin{equation}
\ket{V_{-\hat{\mathbb{I}}}, -1} = (\ket{1}-\ket{2})/\sqrt{2}.
\end{equation}

Note that $V_{{-\hat{\mathbb{I}}}}=H^2$ coincides with the phase point operator $A_{(0,0)}$, and commutes with $H$, $N$, and $N^{-1}$. The eigenvectors of $H$, $N$, and $N^{-1}$ are also eigenvectors of $V_{-\hat{\mathbb{I}}}$. 

\subsubsection*{Eigenvectors of $[V_{\hat{S}}]$ and $[V_{\hat{S}}^{-1}]$}

Both $V_{\hat{S}}$ and $V_{\hat{S}}^{-1}$ have the same eigenvectors, and thus belong to the same reduced conjugacy class. The eigenvectors of $V_{\hat{S}}$ are the stabilizer state, $\ket{V_{\hat{S}},1}=\ket{0}$ and
\begin{eqnarray}
\ket{V_{\hat{S}},\omega_3^2;\gamma,\delta} = \gamma \ket{1}+\delta \ket{2}.
\end{eqnarray}

\subsubsection*{Eigenvectors of $[N]$ and $[N^{-1}]$}

$N$ and $N^{-1}$ belong to different conjugacy classes, but have the same eigenvectors; so they both belong to the same reduced conjugacy class. The eigenvectors of $N$ are:
\begin{eqnarray}
\ket{N,1} = \ket{N^{-1},1} & = & \ket{0} \\
\ket{N, \omega_3^2} = \ket{N^{-1}, \omega_3} & = & \frac{1}{\sqrt{2}} \left( \ket{1}+\ket{2} \right) \equiv \ket{N_+}\\
\ket{N, -\omega_3^2} = \ket{N^{-1}, -\omega_3} & = & \frac{1}{\sqrt{2}} \left( \ket{1}-\ket{2} \right).
\end{eqnarray}

The states $\ket{N, - e^{i \pi/3}}$ and $\ket{N, e^{i\pi/3}}$ are not related to each other by any Clifford operator; both are also eigenvectors of $V_{-\hat{\mathbb{I}}}$.

The set of convex mixtures of these states forms a planar slice of qutrit phase space. This is shown in Figure \ref{fig:B-plane}.

\subsubsection*{Eigenvectors of $[H]$}
The eigenvectors of $H$ are:
\begin{eqnarray}
\ket{H,i} & = & \frac{1}{\sqrt{2}} \left( \ket{1}-\ket{2} \right) \equiv \ket{S}\\
\ket{H,1} & = & \cos \theta \ket{0} + \frac{1}{\sqrt{2}}\sin \theta (\ket{1}+\ket{2}) \equiv \ket{H_+} \\
\ket{H,{-1}} & = & \sin \theta \ket{0} - \frac{1}{\sqrt{2}}\cos \theta (\ket{1}+\ket{2}) \equiv \ket{H_-},
\end{eqnarray}
where $\theta=\frac{1}{2} \arctan \sqrt{2}$.

The eigenvectors $\ket{H,1}$ and $\ket{H,-1}$  are related to each other by the Clifford transformation $\ket{H_1}= V_{\tiny \begin{pmatrix}  1 & 1\\1 & 2 \end{pmatrix}} \ket{H_{-1}}$.

The planar region of qutrit state space spanned by convex mixtures of the eigenvectors of $H$ is shown in Figure \ref{fig:H-plane}.

Note that $H^{-1}$ has eigenvalues $-i,~1, -1$. (It is only because we are ignoring overall phases in our definition of the single-qudit Clifford group that  $H^{-1}$ is in the same conjugacy class as $H$.) 

\subsubsection*{Eigenvectors of $[XV_{\hat{S}}]$ and $[XV_{\hat{S}}^{-1}]$}
$XV_{\hat{S}}$ and $(XV_{\hat{S}})^{-1}$ belong to different conjugacy classes but have the same eigenvectors. The eigenvectors for $XV_{\hat{S}}$ and its inverse were already studied in \cite{HowardVala}, and distillation routines for these states were studied in \cite{CampbellAnwarBrowne}. In terms of $\xi=e^{2\pi i /9}$, these can be written as
\begin{eqnarray}
\ket{XV_{\hat{S}}, \xi^7} & = & \xi^8 \ket{0}+\xi \ket{1}+\ket{2}  \\
\ket{XV_{\hat{S}},\xi^4} & = &  \xi^2 \ket{0}+\xi^7 \ket{1}+\ket{2}  \\
\ket{XV_{\hat{S}},\xi} & = &  \xi^5 \ket{0}+\xi^4 \ket{1}+\ket{2}.
\end{eqnarray}
One can check that the different eigenvectors are related to each other by multiplication by $Z$. So there is one inequivalent magic state which we take to be
\begin{equation}
\ket{XV_{\hat{S}}}= \xi^5 \ket{0}+\xi^4 \ket{1}+\ket{2}.
\end{equation}

\subsection{Qutrit Clifford Eigenstates}
In summary, we find that there are four Clifford-inequivalent non-degenerate qutrit Clifford eigenstates,
\begin{equation}\ket{S},~\ket{H,1},~\ket{N_+}~ \text{and }\ket{XV_{\hat{S}}}.
\end{equation}
In addition there are two Clifford-inequivalent families of degenerate Clifford eigenspaces:
\begin{eqnarray}
    \ket{V_{-\hat{\mathbb{I}}},1;\alpha,\beta} & = & \alpha \ket{0}+{\beta}(\ket{1}+\ket{2}) \\
    \ket{V_{\hat{S}},\omega^2;\gamma,\delta} &= & \gamma\ket{1} + \delta\ket{2}.
\end{eqnarray}
The two degenerate eigenspaces $\ket{V_{-\hat{\mathbb{I}}},1;\alpha,\beta}$ and $\ket{V_{\hat{S}},\omega^2;\gamma,\delta}$ only intersect at points Clifford-equivalent to $\ket{0}$ and $\ket{N_+}$.


The generalized ``hypergraph'' in Figure \ref{fig:venn}, provides a crude graphical summary of these magic states and their properties. (This is reminiscent of the hypergraph construction of \cite{Acin2015}.) The figure consists of several points and lines, contained within intersecting colored regions. In this figure, each (reduced) conjugacy class of the Clifford group is represented by a colored region. Each non-degenerate, inequivalent eigenstate is depicted as a point, and each degenerate family of eigenstates is depicted as a line. An eigenstate of an operator $A$ belonging to $[[A]]$ is contained within the region corresponding to $[[A]]$. Some magic states, such as $\ket{S}$, are contained in more than one region because they are simultaneous eigenvectors of operators belonging to different conjugacy classes.

\begin{figure}
    \centering
    \includegraphics[width=.8\textwidth]{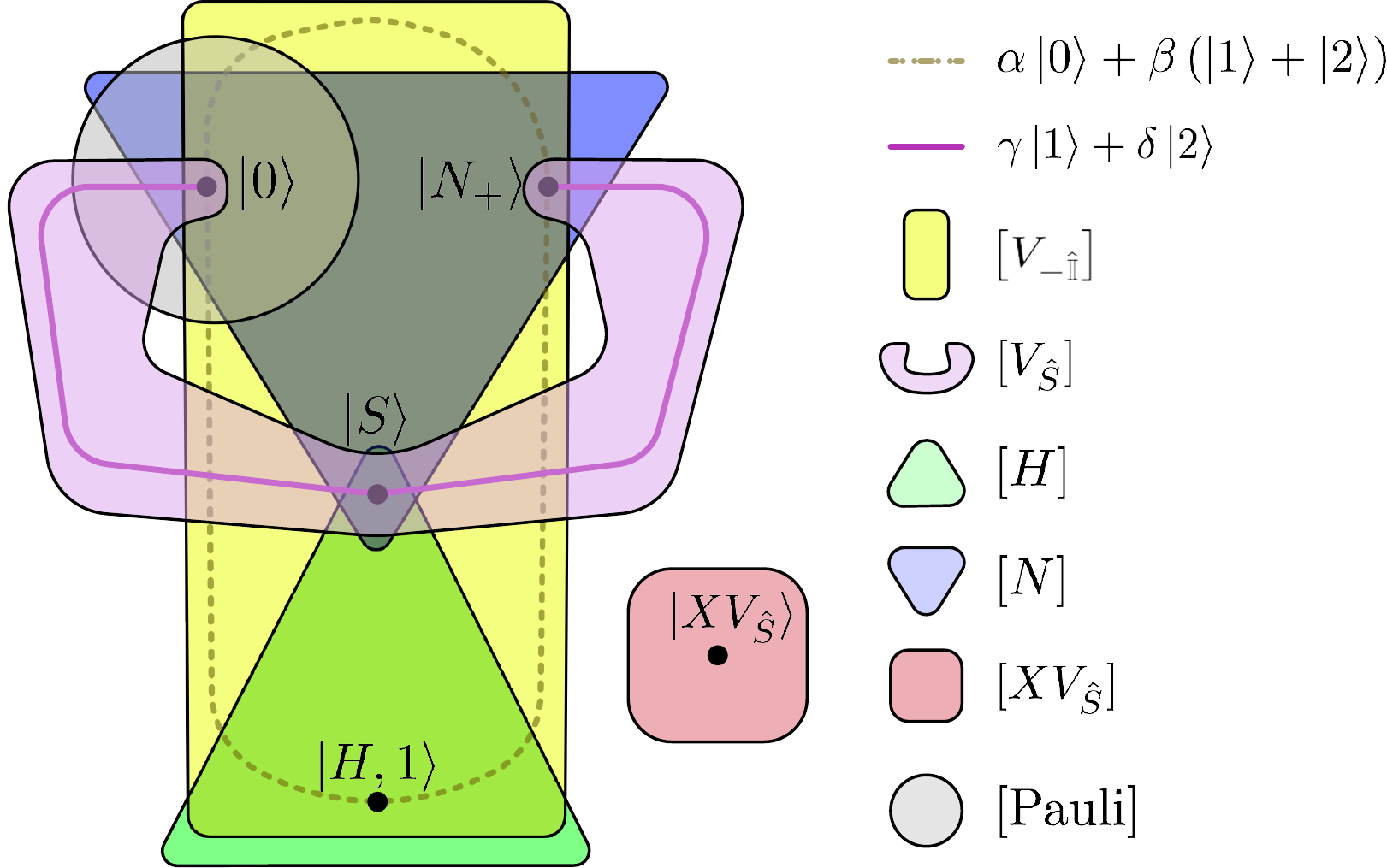}
    \caption{This diagram illustrates all inequivalent qutrit Clifford eigenstates, depicted as points, and their corresponding reduced conjugacy classes, depicted as colored regions. 1-parameter degenerate families of eigenstates are represented by curves.}
    \label{fig:venn}
\end{figure}

In the next sections, we discuss each of the candidate magic states in detail. This discussion is summarized in Table \ref{qutrit-table}. An alternate derivation of these states as qutrit states that maximize the mana is presented in the Appendix.

\begin{table}
\begin{center}
    \begin{tabular}{|l|c|c c|c|c|}
 \hline
    \textbf{State} & \textbf{Eigenvector of} & \textbf{Mana} & & \textbf{Min$(W_\chi(\rho))$} &\textbf{$|$Orbit$|$} \\ [0.5ex]
 \hline
 \hline
 $\ket{S}$ & $H$, $V_{\hat{S}}$, $V_{-\hat{\mathbb{I}}}$, $N$ & $\log \frac{5}{3}$ &$\approx$ 0.51 & -.33 & 9
 \\
 \hline
 $\ket{H,1}$ &  $H$, $V_{-\hat{\mathbb{I}}}$ & $\log \left(\frac{1}{3}+\frac{2}{\sqrt{3}}\right)$ & $\approx$ 0.40 & -.06 & 54 \\
 \hline
 $\ket{N_+}$ & $V_{\hat{S}}$, $V_{-\hat{\mathbb{I}}}$, $N$  & $\log \left(\frac{5}{3}\right)$ & $\approx$ 0.51  & -.17 & 36 \\
 \hline
 $\ket{XV_{\hat{S}}}$ & $XV_{\hat{S}}$ & $\log \left(\frac{1}{3} \left(1+4 \cos \left(\frac{\pi }{9}\right)\right)\right)$ &  $\approx 0.46$ & -.10 & 72 \\
 \hline
\end{tabular}
\caption{\textbf{List of Non-Degenerate Qutrit Clifford Eigenstates:} This table provides a list of non-degenerate qutrit magic states. The first column presents the name of the state. The second column lists the operators that the state is an eigenvector of. The third column lists the mana of the state. The fourth column lists the most negative entry in the state's discrete Wigner function denoted as Min($W_\chi(\rho)$). The last column lists the number of single qudit Clifford-eigenstates that are Clifford equivalent to the given magic state, denoted as $|$Orbit$|$.\label{qutrit-table}}
\end{center}
\end{table}

\subsubsection{The Strange State $\ket{S}$}
The state $\ket{H,i}$ is also known as the strange state $\ket{S}$, and was identified as one of the two states that maximize the mana in \cite{Veitch_2014}. As illustrated from Figure \ref{fig:venn}, the strange state is a simultaneous eigenstate of Clifford unitaries belonging to several different reduced conjugacy classes, with the following eigenvalues:
\begin{eqnarray}
V_{-\hat{\mathbb{I}}}\ket{S} & = & -1 \ket{S} \\
N\ket{S} & = & e^{\pi i /3} \ket{S} \\
V_H \ket{S} & = & i \ket{S} \\
V_{\hat{S}} \ket{S} & = & \omega^2 \ket{S}
\end{eqnarray}

The discrete Wigner function representation of $\ket{S}$ is particularly simple:
\begin{equation}
W_{(u,v)}(\ket{S}\bra{S}) = \begin{cases} -1/3 & (u,v)=(0,0) \\ 1/6 & (u,v) \neq (0,0)\end{cases}.
\label{discrete-wigner-strange}
\end{equation}
It is depicted  in Figure \ref{hiWigner}.

From this Wigner function, we see that the state $\ket{S}$ lies directly above the ``center'' of one facet of the Wigner polytope, and also maximally violates the contextuality inequality of \cite{nature}. Distillation of the $\ket{S}$ state therefore, has the theoretical potential to have the highest threshold to noise of all qutrit magic states, \cite{noiseQudit} although as argued in \cite{qudit-bound-states}, the limit is unattainable by any finite distillation routine. In this sense, it is analogous to $\ket{T}$ states for qubits \cite{MSD,bound}.

The orbit of $\ket{S}$ under the full Clifford group contains 9 states, and its orbit under symplectic rotations is of size $1$.

Using equation \ref{calculate-mana} and the discrete Wigner function \ref{discrete-wigner-strange}, one can directly calculate that the mana of this state is $\log \frac{5}{3}$. This was numerically shown to be maximal in \cite{Veitch_2014}; we present an analytical proof in the Appendix. However, \cite{Veitch_2014} also observed that the Norell state, $\ket{N_+}$, has the same maximal value of the mana, leaving open the question of which of these two qutrit states is most magic. Recently, this question was settled when \cite{wang2018efficiently} showed that the strange state has larger thauma than the Norell state. So, from the perspective of magic as a resource, the strange state is the most magic qutrit state.

For some time, no magic state distillation routine that distills the strange state was known. However, recently it was shown a CSS code based on the ternary Golay code can be used to distill the strange state in \cite{GolaySP}.

\subsubsection{The State $\ket{H,1}$}
The state $\ket{H,1}$ is an eigenstate of the following operators:
\begin{eqnarray}
V_{-\hat{\mathbb{I}}}\ket{H,1} & = & 1 \ket{H,1} \\
V_H \ket{H,1} & = &  \ket{H,1}
\end{eqnarray}

\begin{figure}

\begin{subfigure}[t]{0.45\textwidth}
    \includegraphics[align=c, width=.9\textwidth]{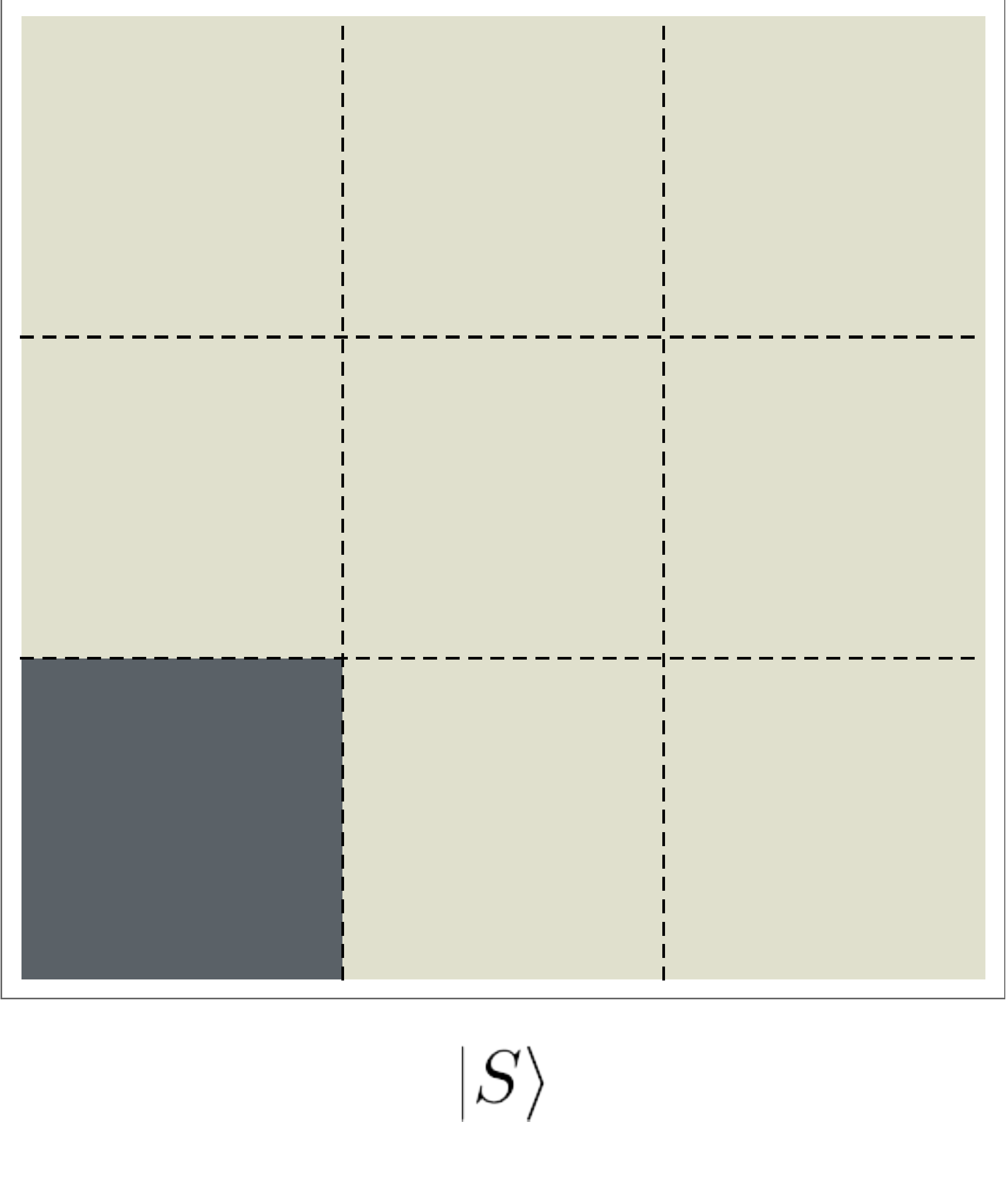}
   
\end{subfigure}
\begin{subfigure}[t]{0.45\textwidth}
    \includegraphics[align=c, width=.9\textwidth]{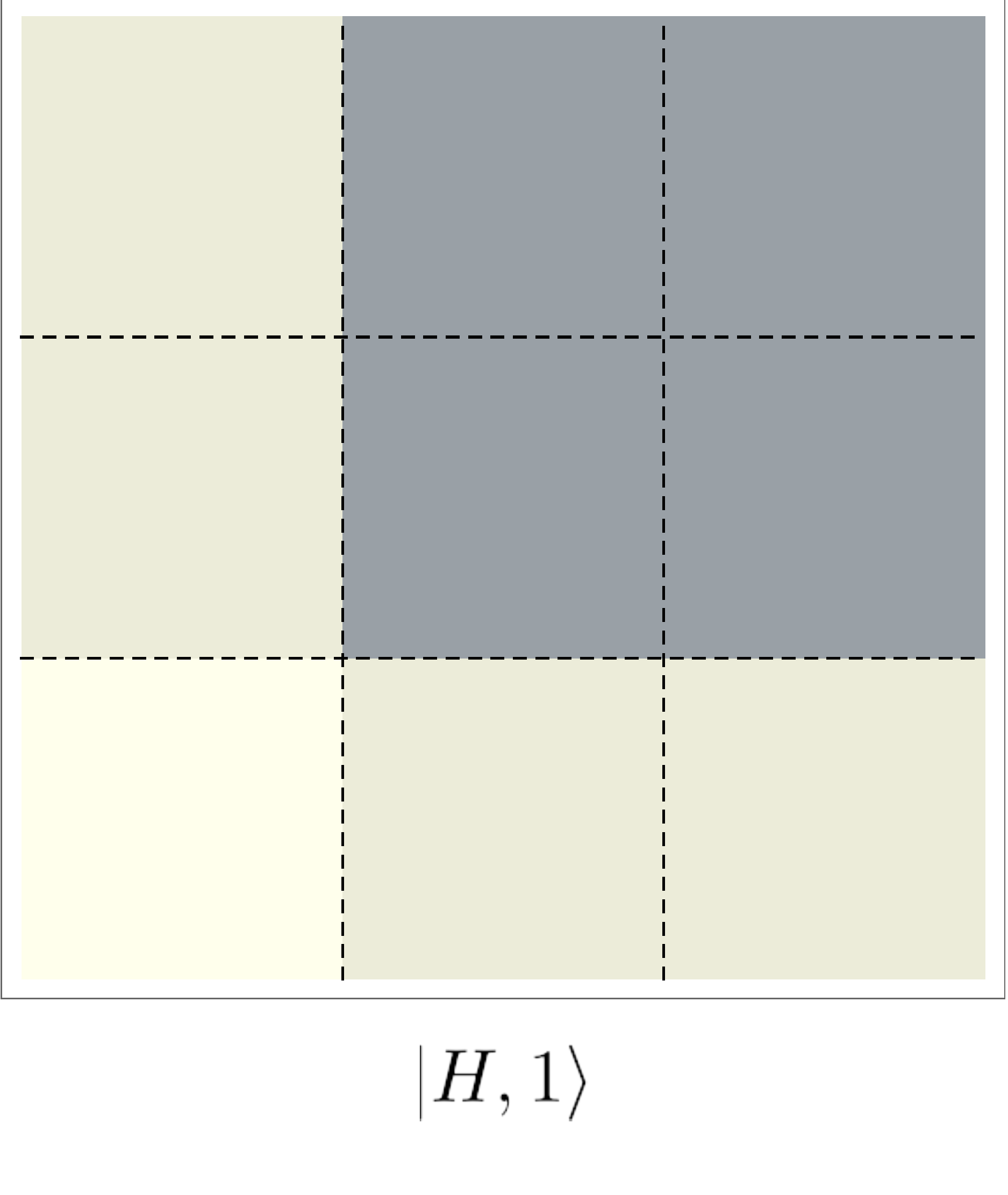}
  
\end{subfigure}

\begin{subfigure}[t]{0.45\textwidth}
    \includegraphics[align=c, width=0.9\textwidth]{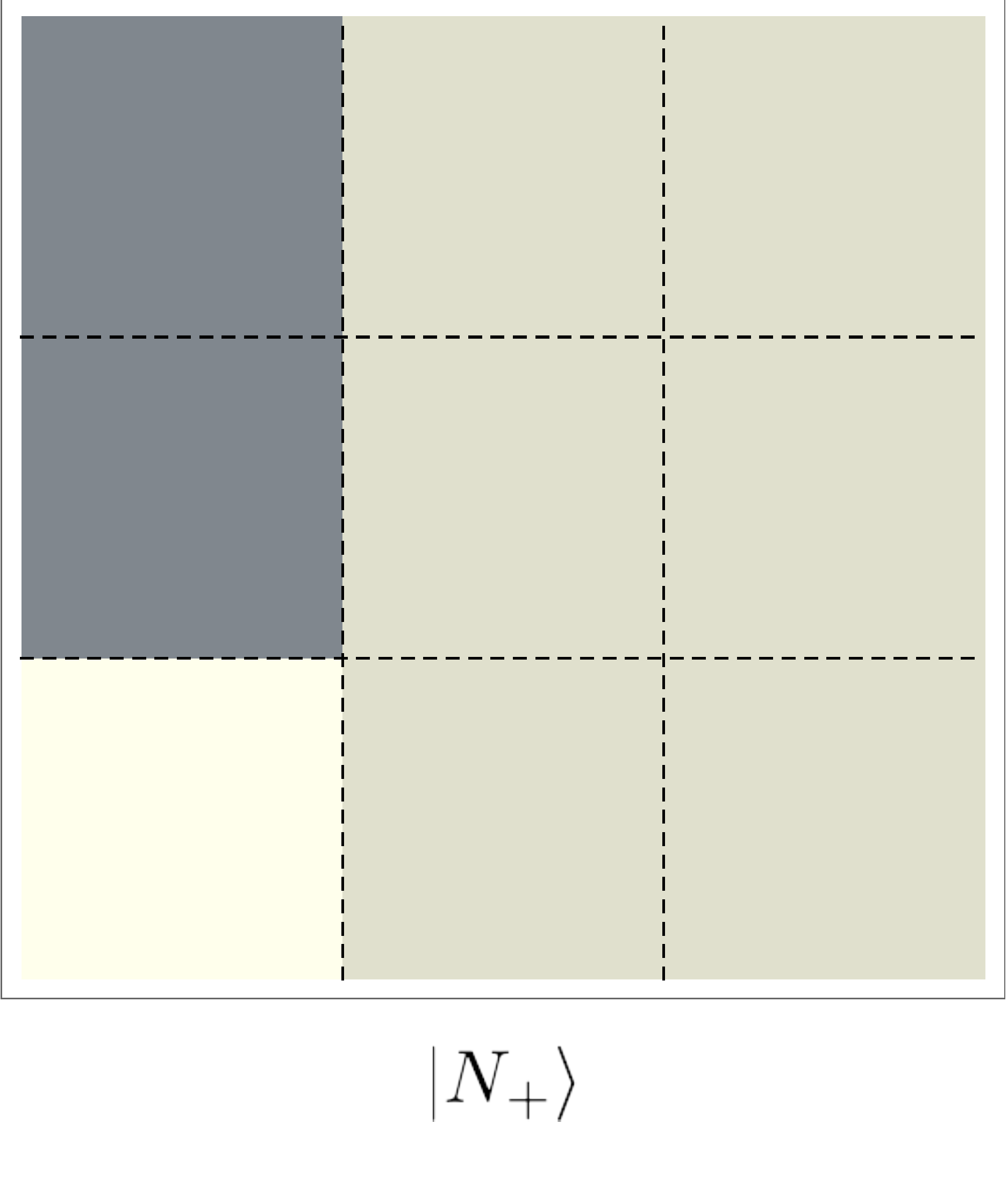}
\end{subfigure}
\begin{subfigure}[t]{0.45\textwidth}
    \includegraphics[align=c, width=0.9\textwidth]{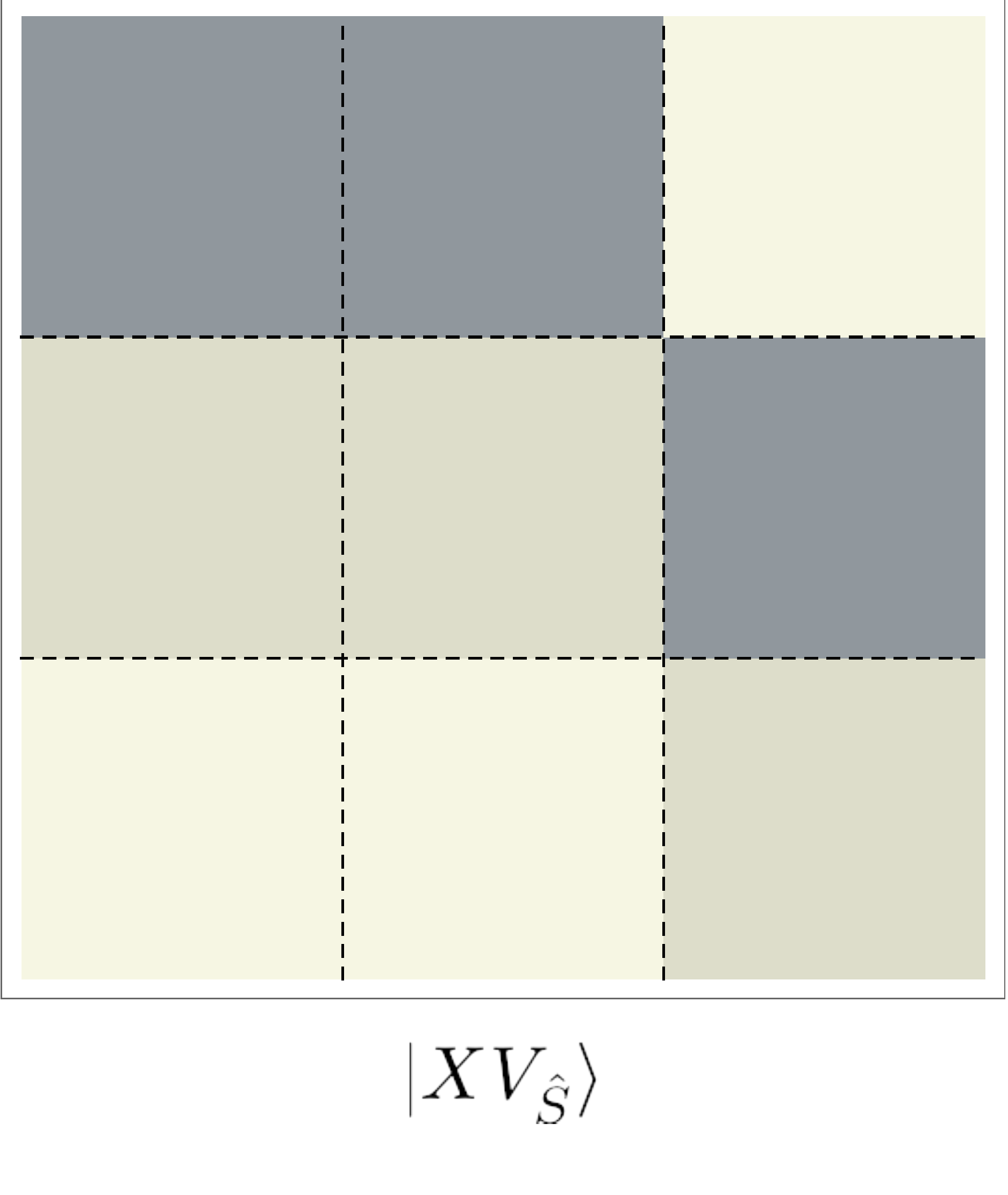}

\end{subfigure}
    \includegraphics[align=c, width=0.07\textwidth]{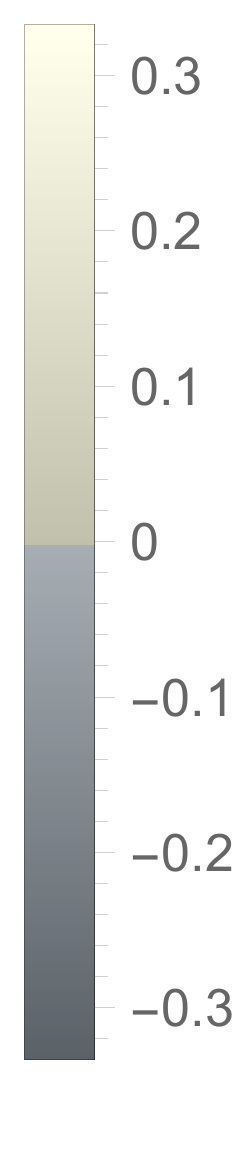}

    \caption{\textbf{Discrete Wigner functions for each of the four non-degenerate qutrit magic states are plotted above.}}
    \label{fig:qutrit-wigner} \label{SWigner}   \label{BWigner}   \label{h1Wigner}
     \label{hiWigner}
\end{figure}

Its discrete Wigner function is given by
\begin{equation}
    W_{(u,v)}(\ket{H,1}\bra{H,1}) = \left(
\begin{array}{ccc}
 W_{(2,0)} & W_{(2,1)} & W_{(2,2)} \\
 W_{(1,0)} & W_{(1,1)} & W_{(1,2)} \\
 W_{(0,0)} & W_{(0,1)} & W_{(0,2)} \\
\end{array}
\right) =
\begin{pmatrix}
 \frac{1}{12} \left(1+\sqrt{3}\right) & \frac{1}{-6-6 \sqrt{3}}
   & \frac{1}{-6-6 \sqrt{3}} \\
 \frac{1}{12} \left(1+\sqrt{3}\right) & \frac{1}{-6-6 \sqrt{3}}
   & \frac{1}{-6-6 \sqrt{3}} \\
 \frac{1}{3} & \frac{1}{12} \left(1+\sqrt{3}\right) &
   \frac{1}{12} \left(1+\sqrt{3}\right) \\
\end{pmatrix}
\end{equation}
and is depicted in Figure \ref{h1Wigner}.

The orbit of $\ket{H,1}$ under the full Clifford group contains 54 states, and its orbit under symplectic rotations contains $6$ states. Its orbit under the quaternion subgroup of symplectic rotations contains $2$ states; its orbit under two of the three  $\mathbb Z_4$ subgroups that do not contain $H$, each contain 2 states. Its orbit under the $\mathbb Z_3$ and $\mathbb Z_6$ subgroups each contain 3 states. 

The mana of $\ket{H,1}$ is $\log \left(\frac{1}{3}+\frac{2}{\sqrt{3}}\right)$. The only known distillation routine for $\ket{H,1}$ states presented in \cite{ACB} and is based on the five-qutrit code. It achieves only a linear reduction in noise,  

\subsubsection{The Norell State $\ket{N_+}$}
The state $\ket{N_+}$ is an eigenvector of the following Clifford operators:
\begin{eqnarray}
V_{-\hat{\mathbb{I}}}\ket{N_+} & = & 1 \ket{N_+} \\
V_N \ket{N_+} & = & -e^{i\pi/3} \ket{N_+}.
\end{eqnarray}

Its discrete Wigner function is:
\begin{equation}
    W_{(u,v)}(\ket{N_+}\bra{N_+}) =
\left(
\begin{array}{ccc}
 -\frac{1}{6} & \frac{1}{6} & \frac{1}{6}
   \\
 -\frac{1}{6} & \frac{1}{6} & \frac{1}{6}
   \\
 \frac{1}{3} & \frac{1}{6} & \frac{1}{6}
   \\
\end{array}
\right)
\end{equation}
This is plotted in Figure \ref{BWigner}.

The orbit of this state under the Clifford group contains 36 states, and its orbit under the set of symplectic rotations contains $4$ states. Its orbit under the quaternion subgroup contains $4$ states, and its orbit under each of the $\mathbb Z_4$ subgroups contain 2 states. Its orbit under the three (of four) $\mathbb Z_6$ subgroups that do not contain $N$ each contain 3 states.  Its orbit under the three (of four) $\mathbb Z_3$ subgroups that do not contain $V_{\hat{S}}$ each also contain $3$ states.

Its mana is $\log \left(\frac{5}{3}\right)$, which is maximal. It was identified as a \textit{Norell state} in \cite{Veitch_2014}. However, as shown in \cite{wang2018efficiently}, its thauma is less than that of the strange state, so it is not the most magic qutrit state.

A distillation routine with linear reduction in noise for the Norell state was found in \cite{Howard}. A substantially better distillation routine was recently discovered, using the ternary Golay code, in \cite{GolaySP}.

\subsubsection{The State $\ket{XV_{\hat{S}} }$}
Its discrete Wigner function is:
\begin{equation}
W_{(u,v)}(\ket{XV_{\hat{S}} } \bra{XV_{\hat{S}} })  = \frac{1}{9}
\left(
\begin{array}{ccc}
 1-2 \cos \left(\frac{\pi }{9}\right) &
   1-2 \cos \left(\frac{\pi }{9}\right) &
   1+2 \cos \left(\frac{2 \pi }{9}\right)
   \\
 1+2 \sin \left(\frac{\pi }{18}\right) &
   1+2 \sin \left(\frac{\pi }{18}\right)
   & 1-2 \cos \left(\frac{\pi }{9}\right)
   \\
 1+2 \cos \left(\frac{2 \pi }{9}\right) &
   1+2 \cos \left(\frac{2 \pi }{9}\right)
   & 1+2 \sin \left(\frac{\pi
   }{18}\right) \\
\end{array}
\right)
\end{equation}
This is plotted in Figure \ref{SWigner}.

The orbit of $\ket{XV_{\hat{S}} }$ under the Clifford group contains $72$ states, and its orbit under the group of symplectic rotations contains $24$ states.

Its mana is $\log \left(\frac{1}{3} \left(1+4 \cos \left(\frac{\pi }{9}\right)\right)\right)$.

Distillation and state-injection schemes for this state were given in \cite{CampbellAnwarBrowne}. It is worth noting that the $\ket{XV_{\hat{S}}}$ state is equatorial, as defined in \cite{CampbellAnwarBrowne}, and a single pure copy of this state can be used to implement the qutrit version of the $\pi/8$ gate (defined and studied in \cite{HowardVala,qutritGate, GLAUDELL201954}) without any chance of error via state injection.

\subsubsection{Degenerate Eigenstates of $V_{-\hat{\mathbb{I}}}$}
Let us parameterize the family of states $\ket{V_{-\hat{\mathbb{I}}}, 1; a,b}$ via $a=\cos \theta$ and $b=e^{i\phi} \sin \theta$.
In terms of these variables, its mana is
\begin{equation}
\begin{split}
\log \Big(\frac{1}{6} & \Big(2 \Big(\left| \sin \theta \left(2 \sqrt{2} \cos \theta \cos \phi+\sin \theta\right)\right|
 + \left| \sin \theta \left(\sin \theta-2 \sqrt{2} \cos \theta \sin \left(\frac{\pi }{6}-\phi \right)\right)\right| \\ &
  +\left| \sin \theta \left(\sin \theta-2 \sqrt{2} \cos \theta \sin \left(\phi +\frac{\pi }{6}\right)\right)\right|
  +1\Big)
  +\left| 3 \cos 2\theta+1\right| \Big)\Big).
\end{split}
\end{equation}
This is plotted in Figure \ref{manaVS}.
All states of the form $\ket{V_{-\hat{\mathbb{I}}}, 1; a,b}$ that maximize the mana are Clifford equivalent to $\ket{N_+}$ which has mana $\log \frac{5}{3}$. There are also local maxima at states Clifford equivalent to $\ket{H,1}$ and the state given by $\theta=\arctan(-\sqrt{2})/2 +\pi/2$, and $\phi=\pi/3$.


Some distillation schemes for the first family of states $\ket{V_{-\hat{\mathbb{I}}}, 1; a,b}$ were studied in \cite{Howard}.

\subsubsection{Degenerate Eigenstates of $V_{\hat{S}}$}
Let us parameterize the family of states $\ket{V_{\hat{S}}, \omega_3^2; \gamma,\delta}$ via $\gamma=\cos \psi$ and $\delta=e^{i\chi}\sin \psi$. In terms of these real variables, its discrete Wigner function is:
\begin{equation}
\frac{1}{3}\left(
\begin{array}{ccc}
 -\sin 2\psi
    \cos (\chi-\pi/3) & \cos ^2\psi & \sin
   ^2\psi \\
 -\sin 2\psi  \cos (\chi+\pi/3) &
   \cos ^2\psi & \sin ^2\psi \\
 \sin 2 \psi   \cos \chi & \cos
   ^2\psi & \sin ^2\psi \\
\end{array}
\right).
\end{equation}
Its mana is
\begin{equation}
\log \left(\frac{1}{3} \left(\left| \sin \left(\frac{\pi }{6}-\chi \right) \sin (2 \psi )\right| + \left| \sin \left(\chi +\frac{\pi }{6}\right) \sin (2 \psi )\right| + \left| \cos \chi \sin (2\psi)\right| +3\right)\right).
\end{equation}
This is plotted in Figure \ref{manaVS}.
The states in the family $\ket{V_{\hat{S}}, \omega_3^2; \gamma, \delta}$ that maximize the mana are Clifford equivalent to either  $\ket{S}=\frac{1}{\sqrt{2}} \left(\ket{1}-\ket{2}\right)$ or $\ket{N_+}=\frac{1}{\sqrt{2}} \left(\ket{1}+\ket{2}\right)$.

\begin{figure}
\centering
\begin{subfigure}[t]{0.49\textwidth}
    \includegraphics[align=c, width=.88\textwidth]{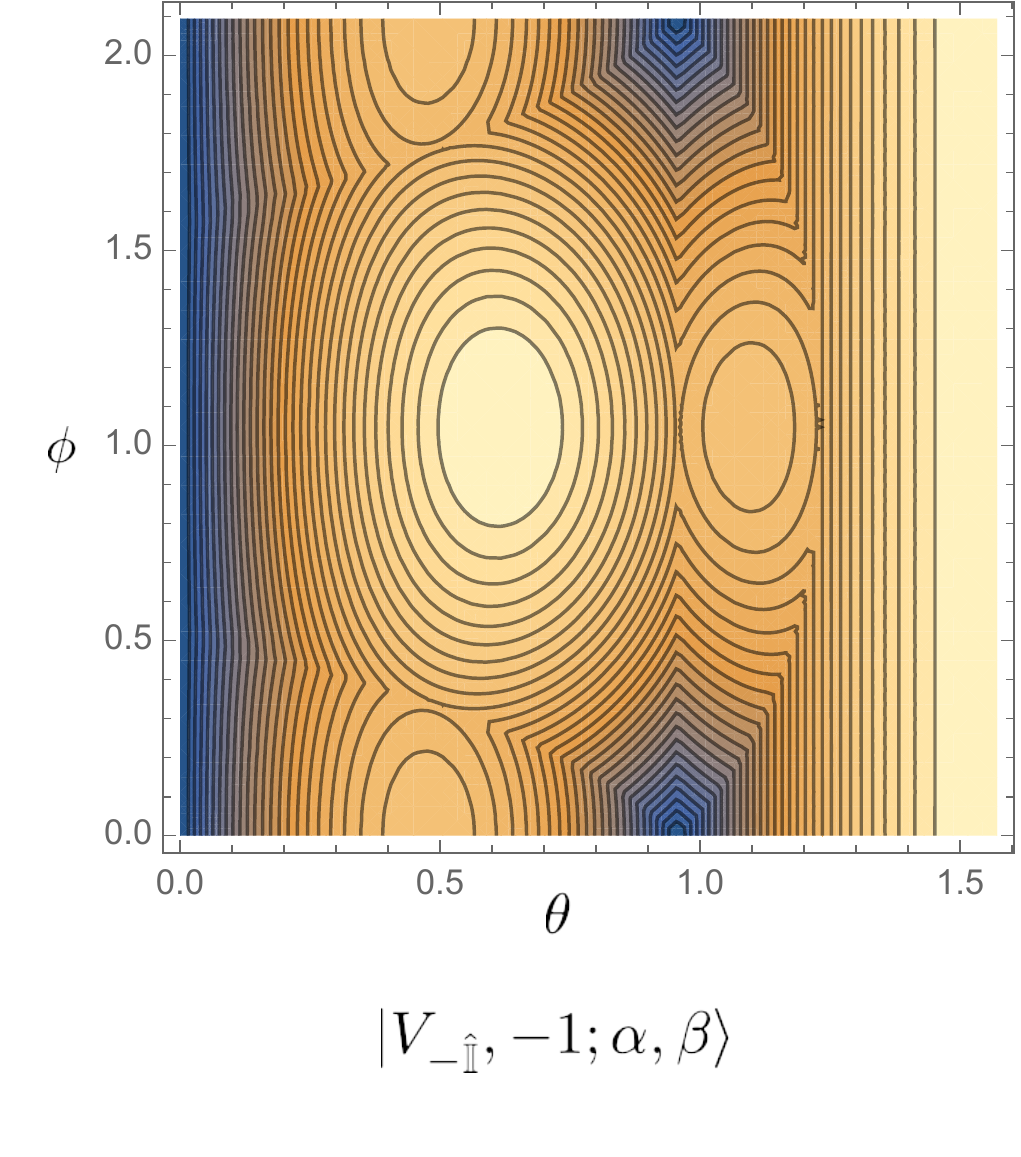}
   
\end{subfigure}
\begin{subfigure}[t]{0.49\textwidth}
    \includegraphics[align=c, width=.88\textwidth]{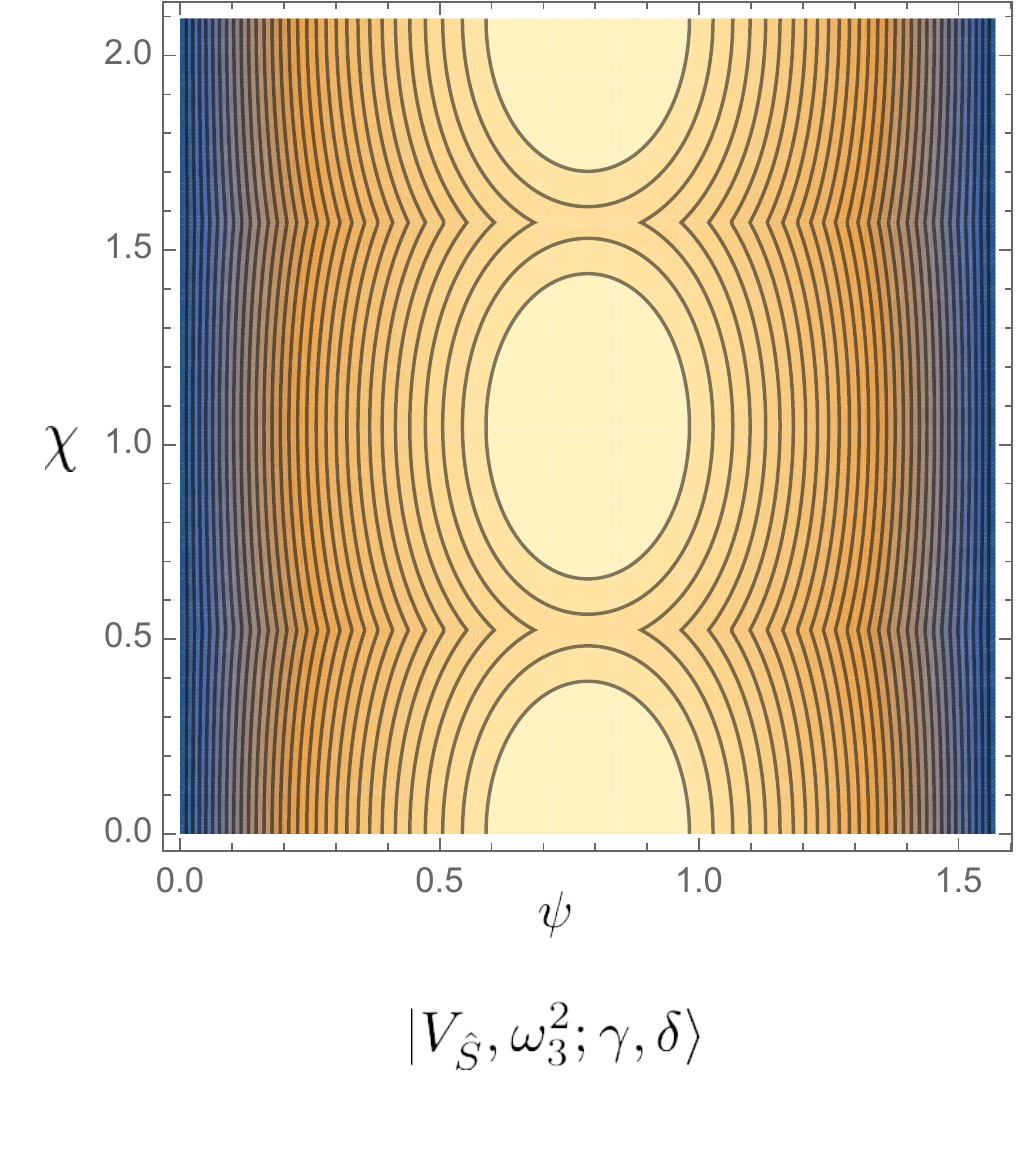}
 \includegraphics[align=c, width=.1\textwidth]{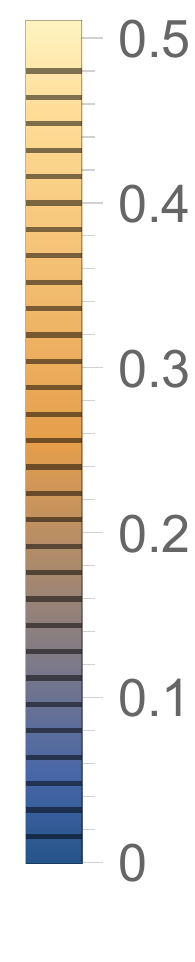}
\end{subfigure}

    \caption{The mana of the two degenerate families of qutrit states, plotted as a function of the angular variables given in the text.}
    \label{degMana} \label{manaVS}
     \label{manaV2}
\end{figure}

\section{Eigenstates of Ququint Clifford Operators}
\label{ququint-states}
We now turn to ququint magic states. We use similar notation as in the previous section, such as $H$ and $V_{\hat{S}}$, for ququint Clifford operators and states. We hope that it is clear from context that all operators and states in this section are ququint operators and states.

\subsection{Conjugacy Classes of the Single-Ququint Clifford group}

Repeating the same computations we carried out for the single-qutrit Clifford group, we find that there are 14 conjugacy classes of the single-ququint Clifford group. These can be grouped into the following 8 reduced conjugacy classes: 
\begin{eqnarray}
~[[I]] & = & \{[I]\} \\
~[[\text{Pauli}] & = & \{[\text{Pauli}]\}, \\
~[[V_{-\hat{\mathbb{I}}}]] &=& \{[V_{-\hat{\mathbb{I}}}]\}, \\
~[[V_{\hat{S}}]] &=& \{[V_{\hat{S}}], ~[V_{\hat{S}}^{-1}]\}, \\
~[[H]] &=& \{[H]\}, \\
~[[A]] &=& \{[V_{\hat{S}}H^2],~[V_{\hat{S}}^2H^2]\},\\
~[[B]] &=& \{[HV_{\hat{S}}],~[(HV_{\hat{S}})^{-1}]\}, \\
~[[XV_{\hat{S}}]] &=& \{[XV_{\hat{S}}], ~[X^2 V_{\hat{S}}], ~[XV_{\hat{S}}^{-1}],~[X^2V_{\hat{S}}^{-1}]\},
 \end{eqnarray}
For notational convenience, we defined the ququint Clifford operators $A=V_{\hat{S}}H^2$ and $B=HV_{\hat{S}}$ in the above list.

The size of each conjugacy class is as follows: $[I]]$ has $1$ element, $[\text{Pauli}]$ has $24$ elements, $[V_{-\hat{\mathbb{I}}}]$ has $25$ elements, $[V_{\hat{S}}]$ has 60 elements, $[H]$ has 750 elements, $[V_{\hat{S}}H^2]$ has 300 elements, $[HV_{\hat{S}}]$ has 500 elements, and $[XV_{\hat{S}}]$ has 120 elements. All conjugacy classes belonging to the same reduced conjugacy class contain the same number of elements. 

Diagonalizing these operators, we find a total of $8$ inequivalent non-degenerate eigenstates, two $1$-parameter families of degenerate eigenstates, and one $2$-parameter family of degenerate eigenstates. We describe each of these below.

\subsubsection*{Eigenstates of $[[V_{-\hat{\mathbb{I}}}]]$}
$V^{(5)}_{-1}$ has two degenerate families of eigenstates:
\begin{eqnarray}
\ket{V^{(5)}_{-1},1;\alpha,\beta,\gamma} & = & \gamma \ket{0}+\alpha (\ket{1}+\ket{4}) + \beta(\ket{2}+\ket{3}) \\
\ket{V^{(5)}_{1},1;\alpha,\beta} & = & \alpha (\ket{1}-\ket{4}) + \beta (\ket{2}-\ket{3}).
\end{eqnarray}

\subsubsection*{Eigenstates of $[[V_{\hat{S}}]]$}
The eigenstates of $V_{\hat{S}}$ are
\begin{eqnarray}
\ket{V_{\hat{S}}, 1} & = & \ket{0} \\
\ket{V_{\hat{S}}, \omega_5^2; \alpha, \beta} & = & \alpha \ket{2}+\beta\ket{3} \\
\ket{V_{\hat{S}}, \omega_5^3;\alpha, \beta} & = & \alpha \ket{1}+\beta\ket{4}.
\end{eqnarray}
The families $\ket{V_{\hat{S}}, \omega_5^2}$ and $\ket{V_{\hat{S}}, \omega_5^3}$ are related to each other by a Clifford transformation.

\subsubsection*{Eigenstates of $[[H]]$}
The eigenstates of $H$ are:
\begin{eqnarray}
\ket{H,-1} & = & (10-2\sqrt{5})^{-1/2} \left( (1-\sqrt{5}) \ket{0} + \ket{1} + \ket{2} + \ket{3} + \ket{4} \right) \\
\ket{H, \pm i} & = & \frac{1}{2} \left( \sqrt{1\mp \chi} \left(\ket{1}-\ket{4}\right) + \sqrt{1\pm \chi}\left(\ket{2}-\ket{3}\right) \right)\\
\ket{H,1; \alpha,\beta} & = & (1+\sqrt{5})/2 (\alpha-\beta)\ket{0} + \alpha \left(\ket{1}+\ket{4} \right) + \beta\left( \ket{2}+\ket{3} \right)
\end{eqnarray}
where $\chi=\sqrt{\frac{1}{10} \left(5+\sqrt{5}\right)}$. $\ket{H,\pm i}$ are related to each other via a Clifford transformation. $\ket{H,-1}$ can be mapped to a member of $\ket{H,1}$ via a Clifford transformation.

\subsubsection*{Eigenstates of $[[A]]$}
The eigenstates of $A=V_{\hat{S}}H^2$ are:
\begin{eqnarray}
\ket{A, 1} & = & \ket{0} \\
\ket{A,\pm\omega_5^2} & = & \frac{1}{\sqrt{2}} (\ket{2}\pm\ket{3}) \\
\ket{A,\pm \omega_5^3} & = & \frac{1}{\sqrt{2}} (\ket{1}\pm\ket{4}).
\end{eqnarray}
$\ket{A_+} \equiv \ket{A,\omega_5^2}$ and $\ket{A, \omega_5^3}$ are related to each other by a Clifford transformation. $\ket{A_-} \equiv \ket{A,-\omega_5^2}$ and $\ket{A, -\omega_5^3}$ are also related to each other by a Clifford transformation.

\subsubsection*{Eigenstates of $[[B]]$}
The reduced conjugacy class of $B=HV_{\hat{S}}$ also contains the operator $B'=KBK^{-1}$, where $K=XV_{\begin{pmatrix} 1 & 2 \\ 2 & 0 \end{pmatrix}}$. The unnormalized eigenvectors of $B'$, which are all real and simpler to write down than the eigenvectors of $B$, can be presented as:
\begin{eqnarray}
\ket{B', -1} & = & \frac{1}{2} \left(3+\sqrt{5}\right)\ket{0}+\ket{1}+\ket{2}+\ket{3}+\ket{4} \\
\ket{B',e^{\frac{\pm 2 \pi  i}{3}}} & = & \frac{1}{4} \eta_\pm  (\ket{1}-\ket{4}) +\ket{2}-\ket{3} \\
\ket{B',-e^{\frac{\pm 2 \pi  i}{3}} } & = & \kappa_\pm \ket{0} -\kappa_\pm^2/4 (\ket{1}+\ket{4}) + \ket{2} + \ket{3}.
\end{eqnarray}
where $\eta_\pm=\left(\mp\sqrt{30-6 \sqrt{5}}+\sqrt{5}-3\right)$ and $\kappa_\pm=\frac{1}{2} \left(\pm\sqrt{6 \left(5+\sqrt{5}\right)}-\sqrt{5}-3\right)$. Of these states, $\ket{B',-e^{\frac{\pm 2 \pi  i}{3}} }$ are equivalent to each other by a Clifford transformation, and $\ket{B',e^{\frac{\pm 2 \pi  i}{3}}}$ are equivalent to each other by a Clifford transformation.

In Figure \ref{ququint-wigner}, we plot Wigner functions for the eigenvectors of $B$ not $B'$, because the symmetry of $B$ is easier to visualize in discrete phase space.

\subsubsection*{Eigenstates of $[[XV_{\hat{S}}]]$}
The reduced conjugacy class $[[XV_{\hat{S}}]]$ includes the conjugacy classes $[XV_{\hat{S}}]$, $[X^2V_{\hat{S}}]$, $[XV_{\hat{S}}^{-1}]$ and $[X^2V_{\hat{S}}^{-1}]$. Its eigenstates (which were first found in \cite{campbell2014enhanced, HowardVala}), are:
\begin{eqnarray}
\ket{XV_{\hat{S}}, 1} & = & \ket{0} + \ket{1} + \omega_5^3 \ket{2}+ \ket{3} + \omega_5^2 \ket{4} \\
\ket{XV_{\hat{S}},\omega_5^n} & = & (Z^\dagger)^{n}\ket{XV_{\hat{S}}, 1},~\text{ for }n=1,\ldots,4.
\end{eqnarray}
All of these eigenstates are related to each other by a Clifford transformation.

\subsection{Ququint Clifford Eigenstates}
In summary, we have nine inequivalent non-degenerate eigenstates (including $\ket{0}$), $3$ one-parameter families of degenerate states, and $1$ two-parameter family of degenerate states. These are shown in Figure \ref{ququintVenn}. Colored regions correspond to reduced conjugacy classes, and eigenstates of conjugacy classes are contained in their corresponding regions, as in Figure \ref{fig:venn}. Degenerate families of eigenstates are shown as lines.

The only intersections of $\ket{V_{\hat{S}}, \omega_5^2}$ and $\ket{H,1}$ are Clifford-equivalent to $\ket{0}$. The only intersections of $\ket{V_{\hat{S}}, \omega_5^2}$ and $\ket{V_{-\hat{\mathbb{I}}}, \omega_5^2}$ are Clifford-equivalent to $\ket{A_2}$. $\ket{H,1}$ and $\ket{V_{-\hat{\mathbb{I}}}}$ have no intersections.

We plot the discrete Wigner function of the eight non-degenerate non-stabilizer states in Figure \ref{ququint-wigner}. For each non-degenerate state, we list the reduced conjugacy classes it is an eigenvector of, the mana, and the size of its orbit under the Clifford group in Table \ref{ququint-table}.  We plot the mana for each of the one parameter families of degenerate eigenstates in Figure \ref{degenerate5}.
\begin{figure}
\centering
\includegraphics[width=.9\textwidth]{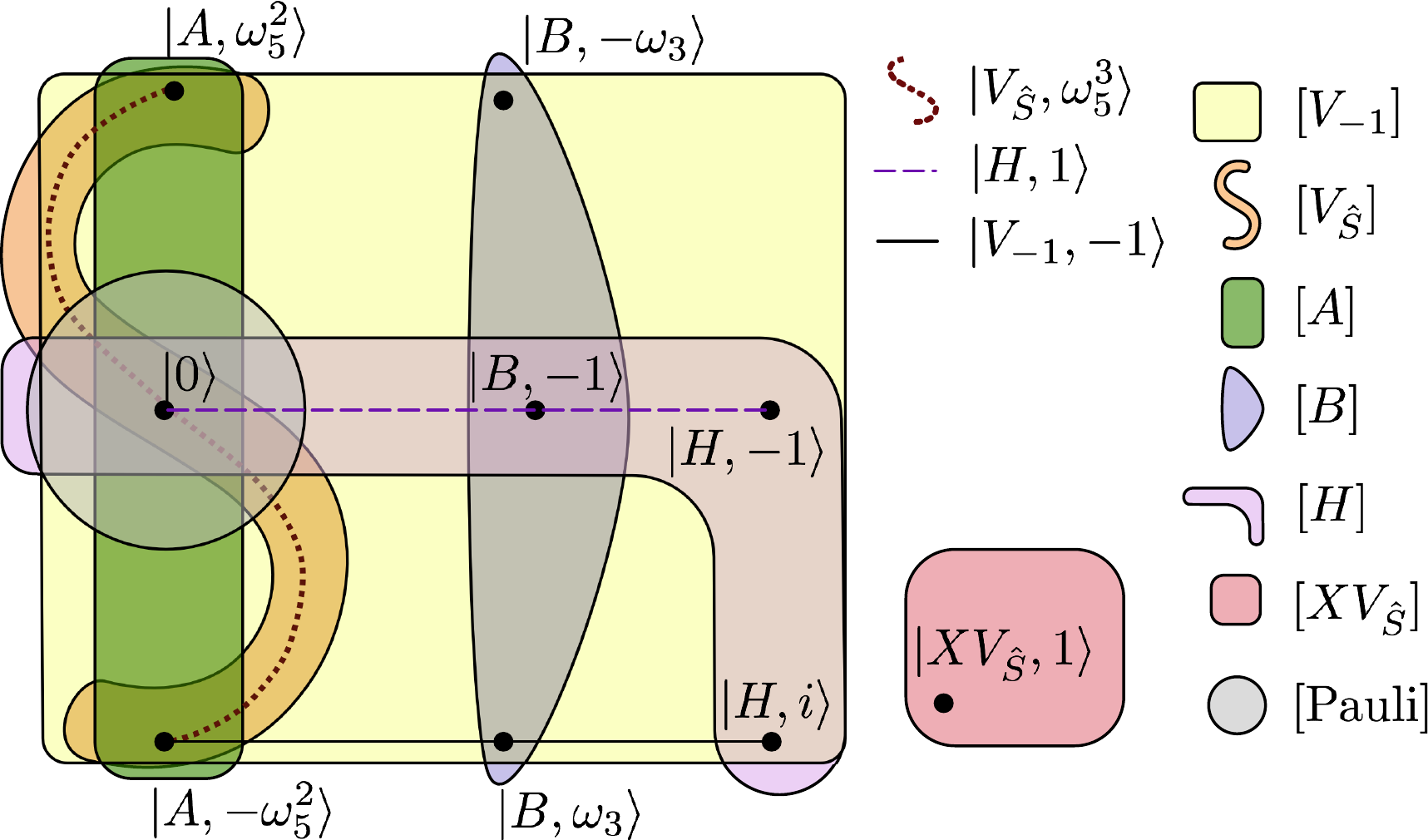}
\caption{This diagram illustrates all inequivalent ququint Clifford eigenstates, depicted as points, and their corresponding reduced conjugacy classes, depicted as colored regions. 1-parameter degenerate families of eigenstates are represented by curves. The 2-parameter degenerate family of states $\ket{V_{-\hat{\mathbb{I}}},1}$ is not pictured. \label{ququintVenn}}
\end{figure}

Based on a numerical search, we find that the maximal mana for any ququint state is
\begin{equation}
\mathcal M_5 = \sinh ^{-1}\left(3+\sqrt{5}\right)-\log (5),
\end{equation}
which is attained by the states Clifford equivalent to $\ket{B',-e^{\frac{2 \pi  i}{3}} }$. We are not sure if there are any other states which have the same mana. Therefore we conjecture that $\ket{B',-e^{\frac{\pm 2 \pi  i}{3}} }$ is the most magic ququint state. The most symmetric Clifford eigenstate, which we define as that eigenstate with the smallest orbit under the Clifford group is $\ket{B,-1}$. Unlike the qutrit case, the most magic state is not the most symmetric. Also, as can be seen from this Figure, there is no ququint Clifford eigenstate with exactly one negative entry in the Wigner function.

Magic state distillation routines for the state $\ket{XV_{\hat{S}},1}$ were constructed in \cite{CampbellAnwarBrowne, campbell2014enhanced}. To our knowledge, magic state distillation routines for the other ququint magic states have not yet been constructed.

\begin{figure}

\begin{subfigure}[t]{0.3\textwidth}
    \includegraphics[align=c, width=.9\textwidth]{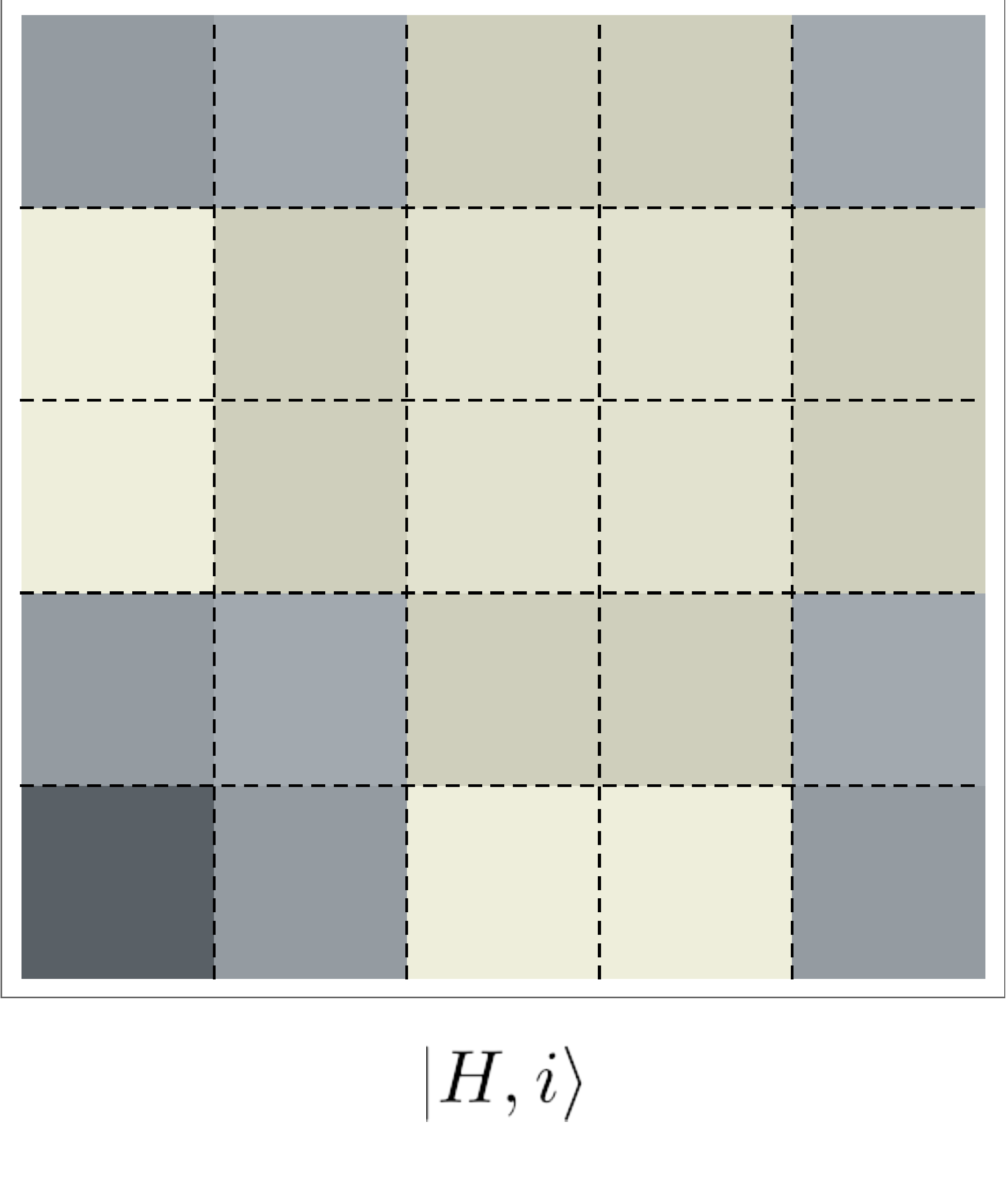}
\end{subfigure}
\begin{subfigure}[t]{0.3\textwidth}
    \includegraphics[align=c,width=.9\textwidth]{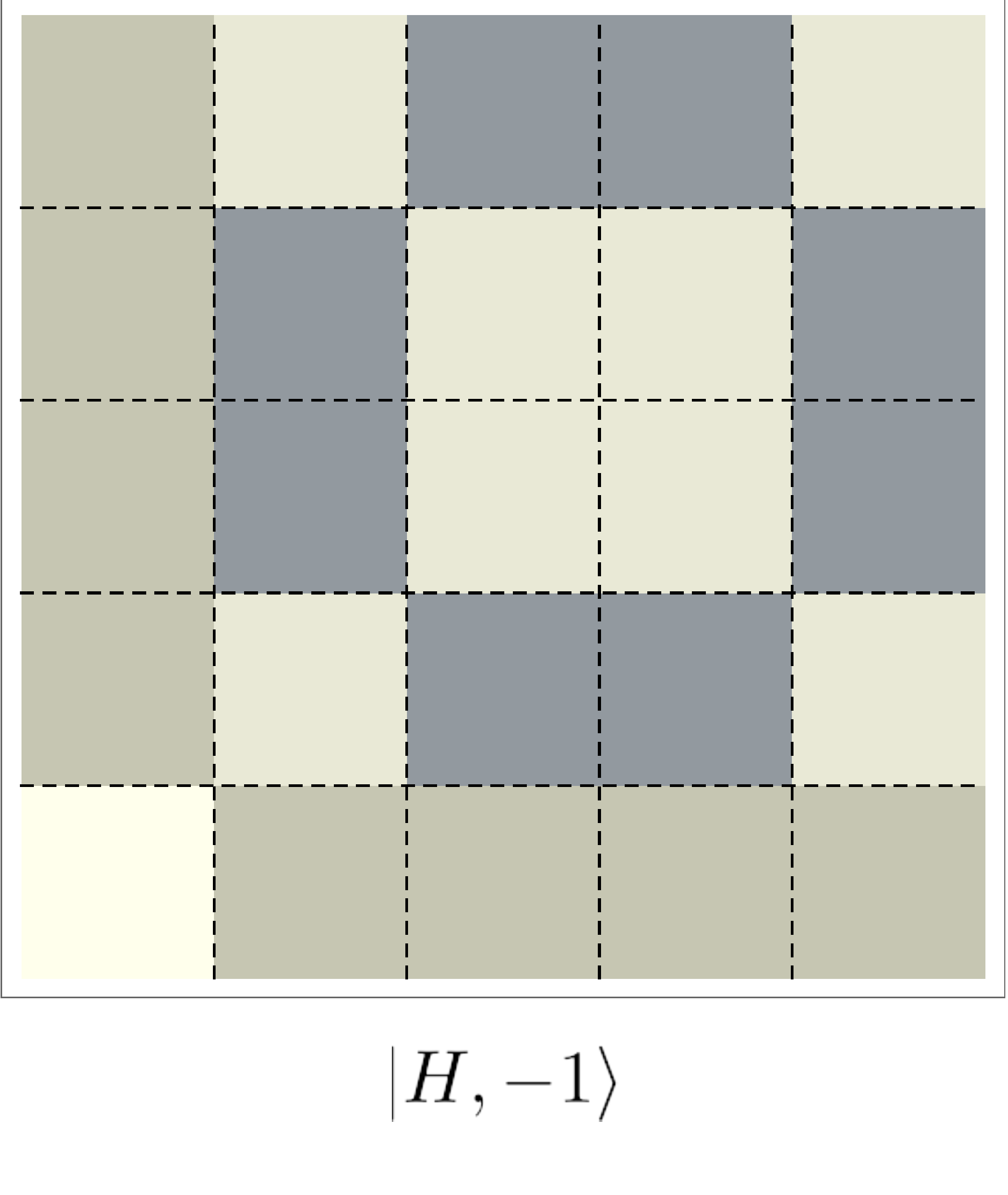}
\end{subfigure}
\begin{subfigure}[t]{0.3\textwidth}
    \includegraphics[align=c, width=0.9\textwidth]{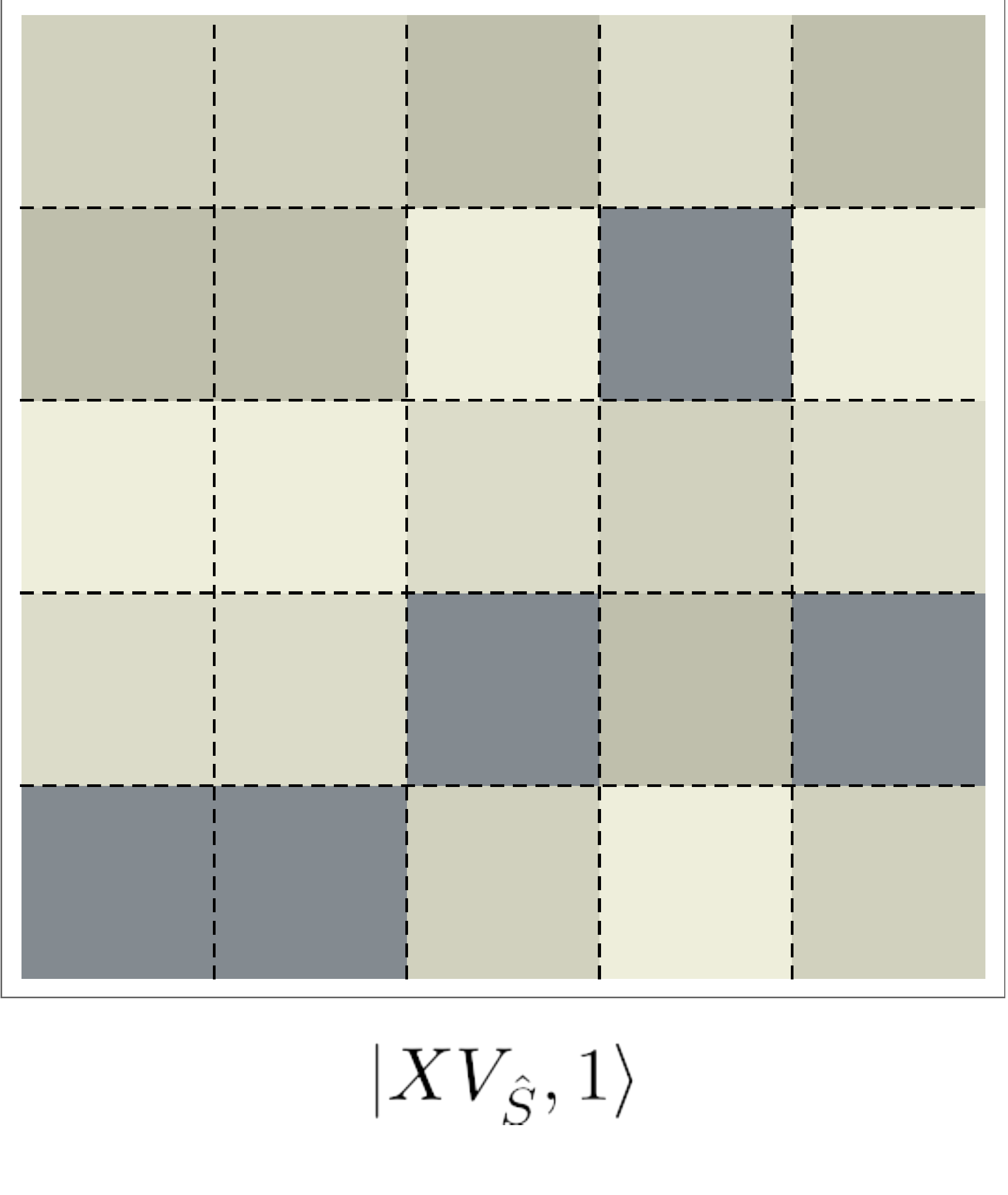}
\end{subfigure}

\begin{subfigure}[t]{0.3\textwidth}
    \includegraphics[align=c, width=0.9\textwidth]{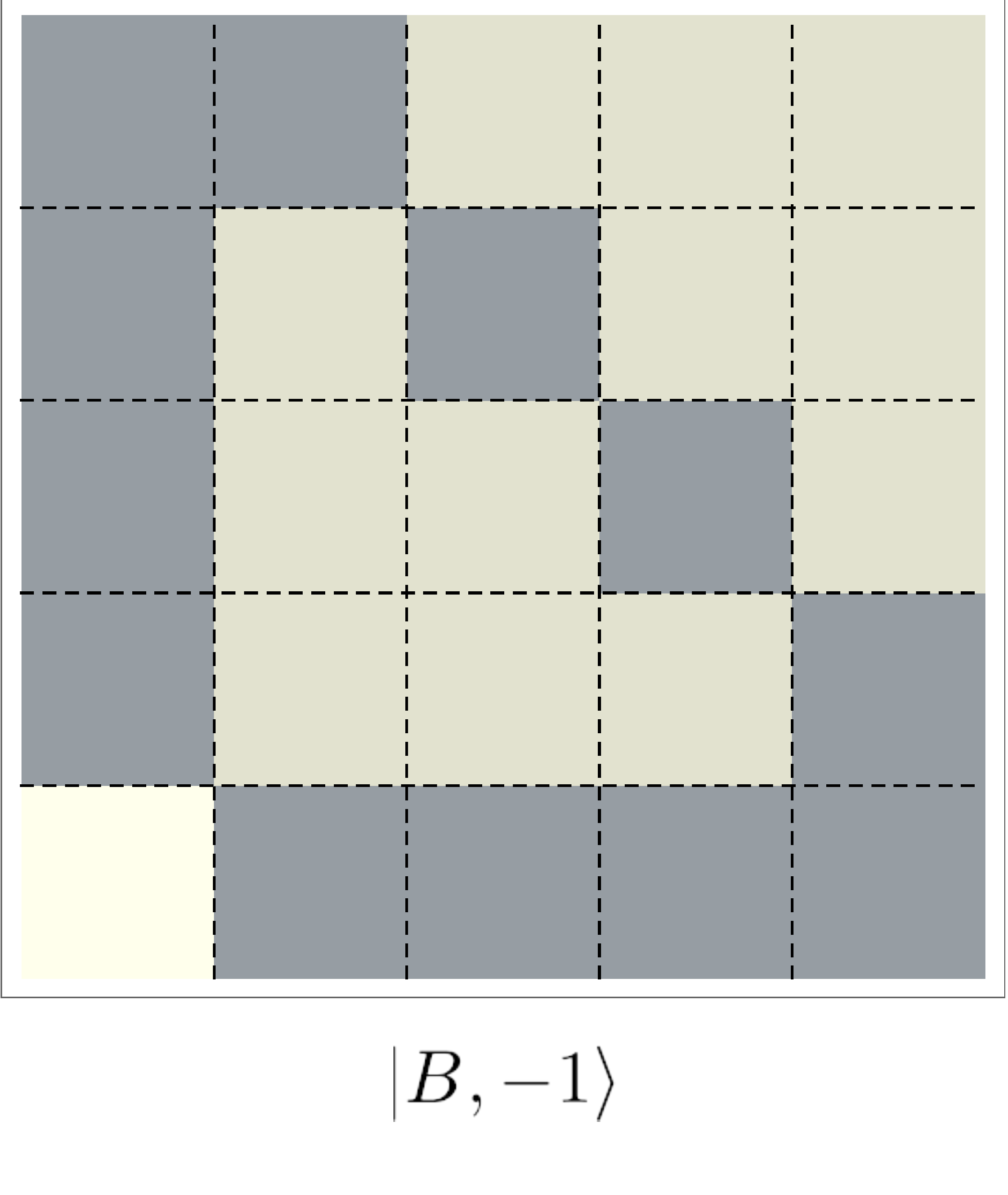}
\end{subfigure}
\begin{subfigure}[t]{0.3\textwidth}
    \includegraphics[align=c,width=0.9\textwidth]{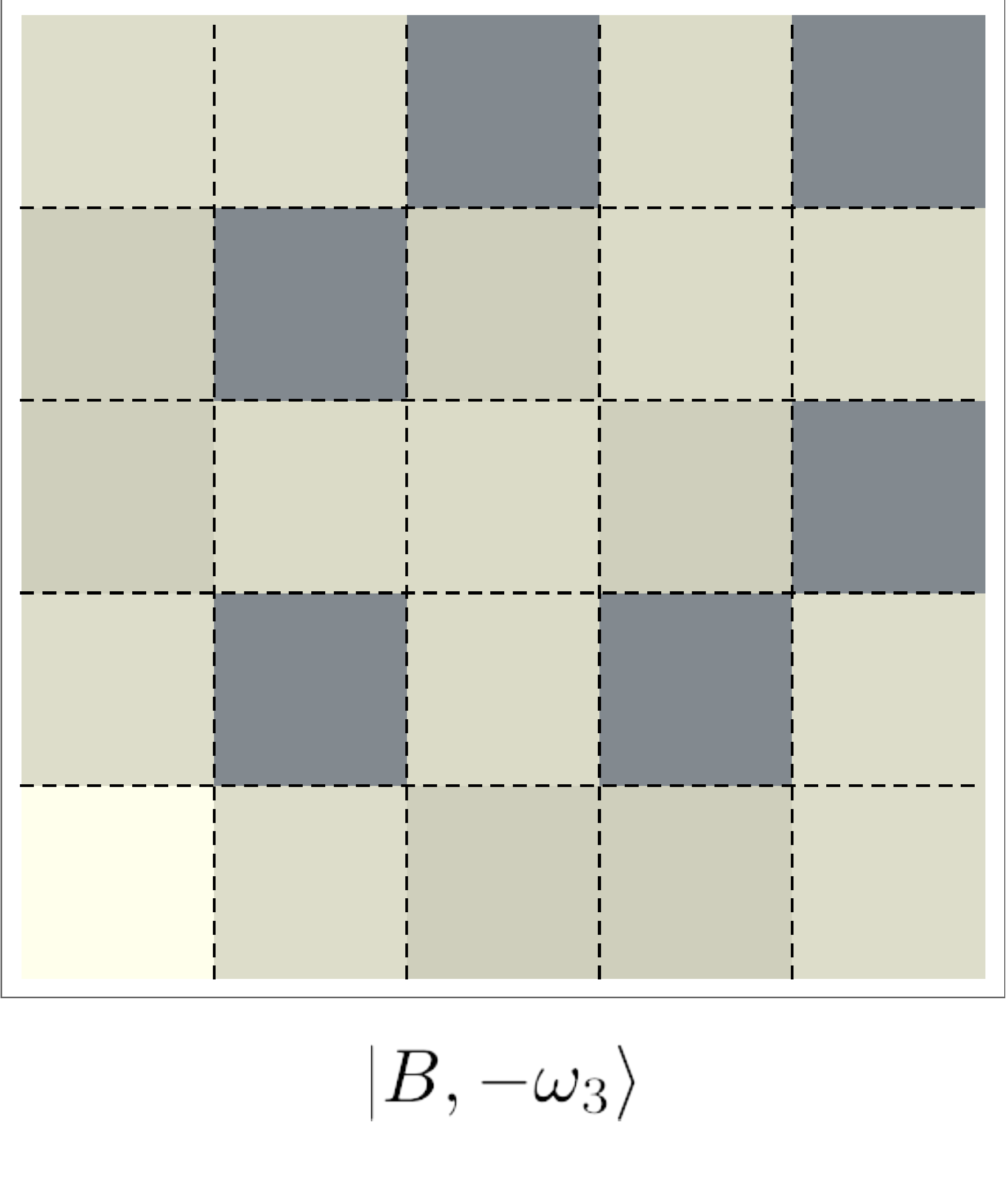}

\end{subfigure}
\begin{subfigure}[t]{0.3\textwidth}
    \includegraphics[align=c,width=0.9\textwidth]{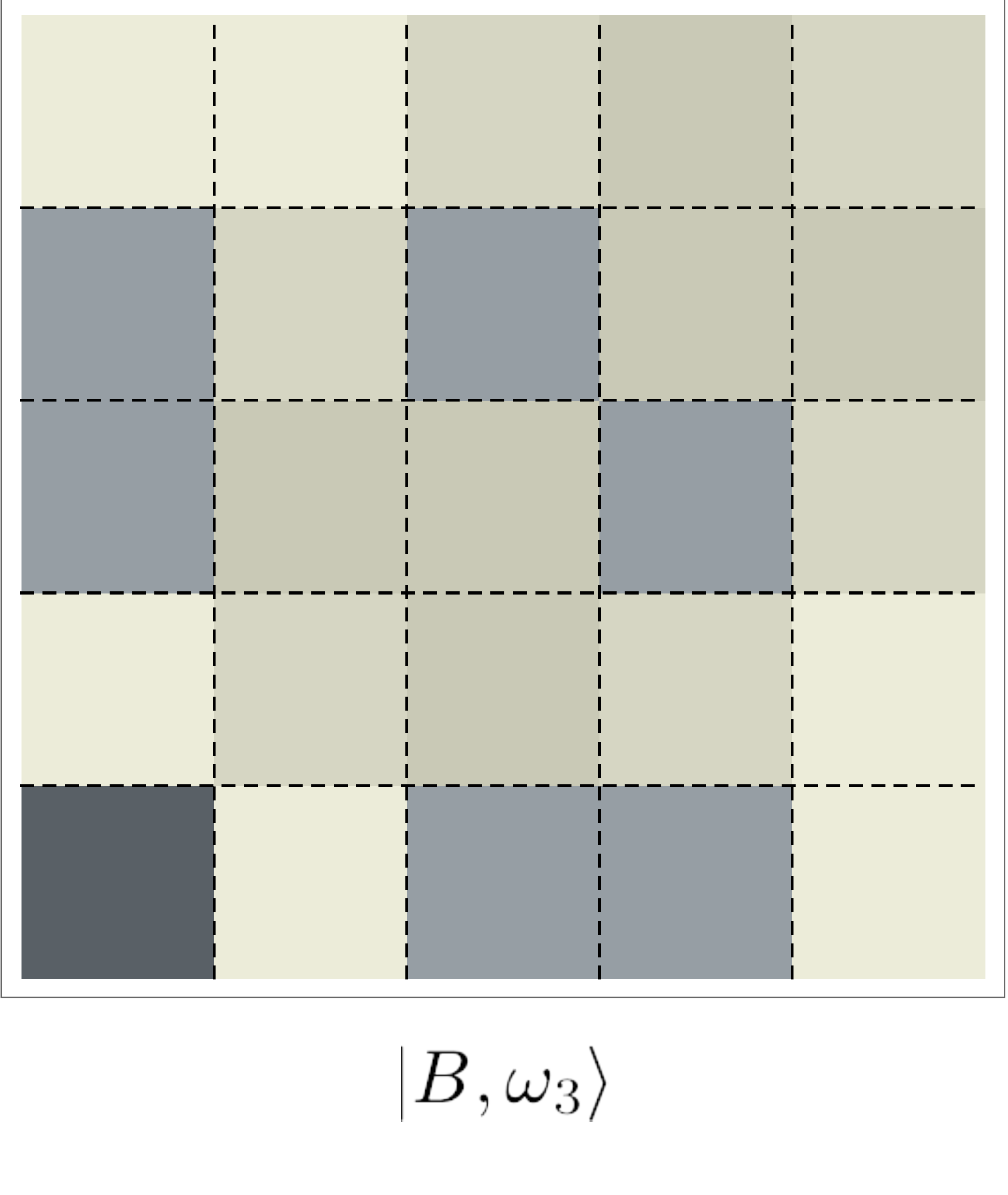}

\end{subfigure}
\includegraphics[align=c,width=0.047\textwidth]{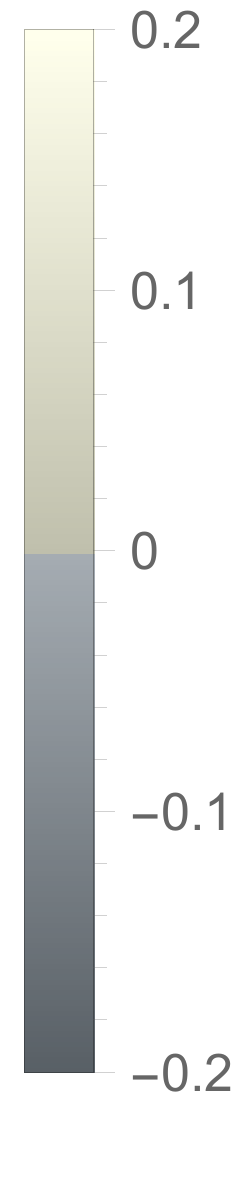}

\centering
\begin{subfigure}[t]{0.3\textwidth}
    \includegraphics[align=c,width=0.9\textwidth]{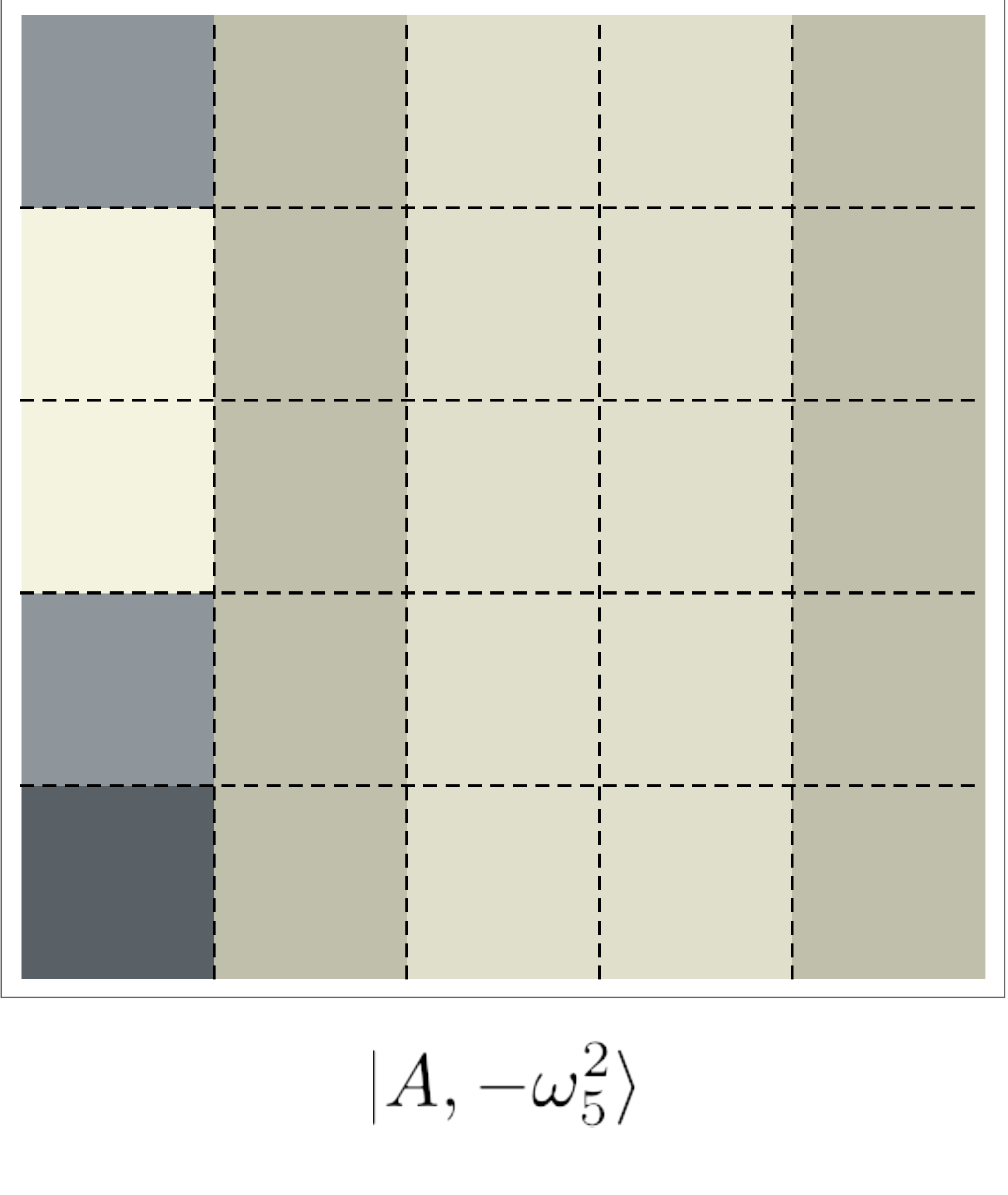}
 
\end{subfigure}
\begin{subfigure}[t]{0.3\textwidth}
    \includegraphics[align=c,width=0.9\textwidth]{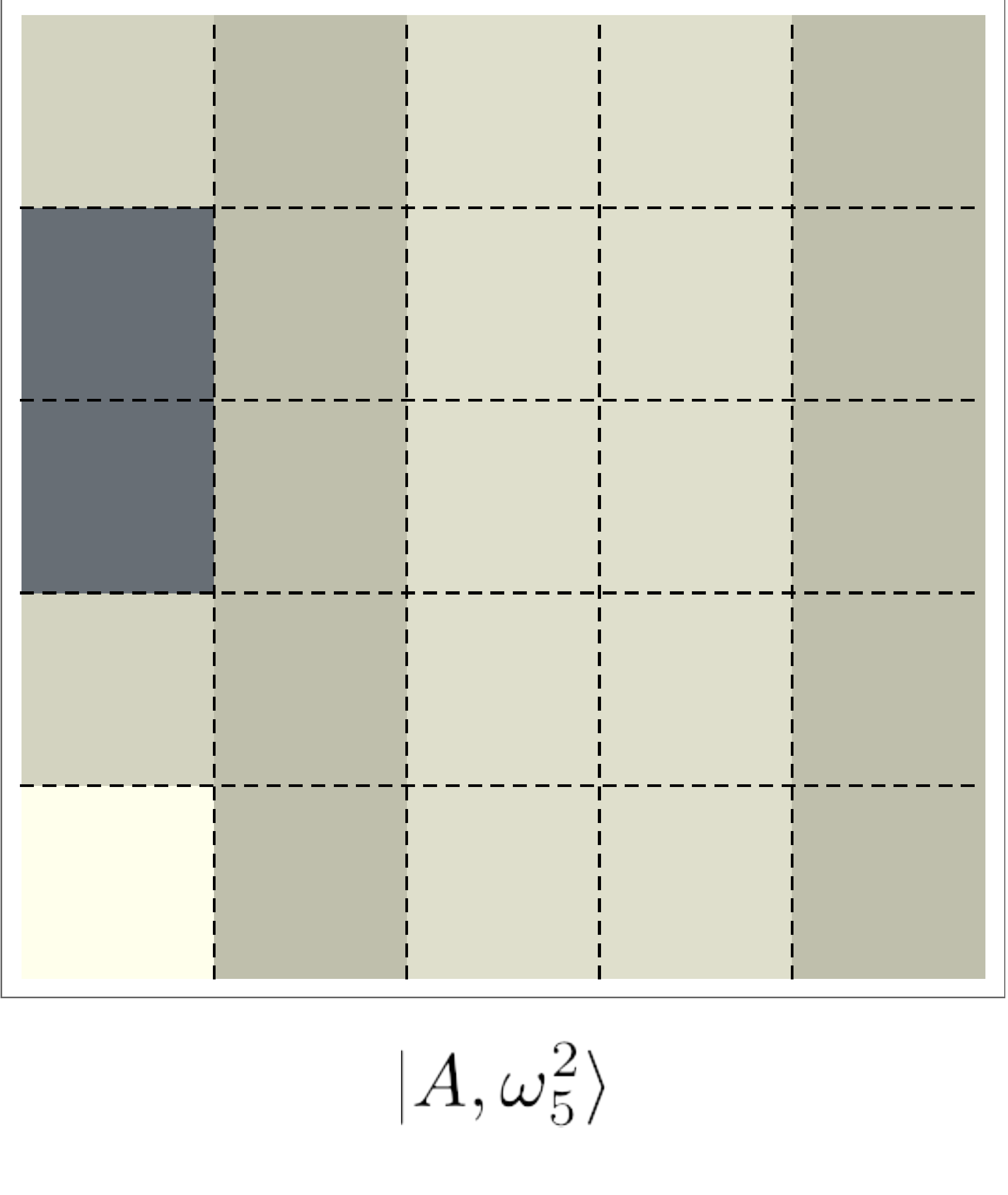}

\end{subfigure}

    \caption{\textbf{Discrete Wigner functions for each of the 8 non-degenerate ququint Clifford eigenstates}}
    \label{ququint-wigner}
\end{figure}

\begin{table}[h!]
  \begin{center}

    \begin{tabular}{|l|c|c c|c|c|}
    \hline
      \textbf{State} & \textbf{Eigenvector of } & \textbf{Mana} & & \textbf{Min$(W_\chi(\rho))$} &\textbf{$|$Orbit$|$} \\ [0.5ex]
       \hline \hline
      $\ket{H,i}$ & $H$, $V_{-\hat{\mathbb{I}}}$ & $\log \frac{1}{5} \left(2 \sqrt{5+2 \sqrt{5}}+3\right)$ &$\approx$ .605 & -.20 &$750$
      \\
      $\ket{H,-1}$ & $H$, $V_{-\hat{\mathbb{I}}}$ & $\log  \frac{9}{5} $ & $\approx$ .588  & -.05 & 375
  \\
  $\ket{B, -1}$ & $B$, $H$, $V_{-\hat{\mathbb{I}}}$ & $\log \left(\frac{1}{5}+\frac{4}{\sqrt{5}}\right)$ & $\approx$ .688 & -.04 & 250  \\
  $\ket{B,-e^{\frac{2 \pi  i}{3}} }$ & $B$, $V_{-\hat{\mathbb{I}}}$ & $\sinh ^{-1}\left(3+\sqrt{5}\right)-\log (5)$ & $\approx$ .748 & -.09 & 500 \\
  $\ket{B,e^{\frac{2 \pi  i}{3}}}$ & $B$, $V_{-\hat{\mathbb{I}}}$ & $\log \frac{\sqrt{15+6 \sqrt{5}}+4}{5}$ & $\approx$ .624 & -.20 & 500 \\
  $\ket{A,-\omega_5^2}$ & $A$, $V_{\hat{S}}$, $V_{-\hat{\mathbb{I}}}$ & $\log \left(\frac{6}{5}+\frac{1}{\sqrt{5}}\right)$ & $\approx$ .499 & -.20 & 300 \\
  $\ket{A,\omega_5^2}$ & $A$, $V_{\hat{S}}$, $V_{-\hat{\mathbb{I}}}$ & $\log \left(\frac{6}{5}+\frac{1}{\sqrt{5}}\right)$ & $\approx$ .499 & -.16 & 300 \\
  $\ket{XV_{\hat{S}},1}$ & $XV_{\hat{S}}$ & $\log \left(1+\frac{2}{\sqrt{5}}\right)$ & $\approx$ .634 & -.09 & 600 \\
  \hline
    \end{tabular}

   \caption{\textbf{List of Non-Degenerate Ququint Clifford Eigenstates}-- \label{ququint-table}  This table provides a list of non-degenerate ququint magic states. The first column presents the name of the state. The second column lists the operators that the state is an eigenvector of. The third column lists the mana of the state. The fourth column lists the most negative entry in the state's discrete Wigner function denoted as Min($W_\chi(\rho)$). The last column lists the number of single qudit Clifford-eigenstates that are Clifford equivalent to the given magic state, denoted as $|$Orbit$|$.}
  \end{center}
\end{table}

\begin{figure}
\centering
\begin{subfigure}[t]{0.49\textwidth}
    \includegraphics[align=c,width=.9\textwidth]{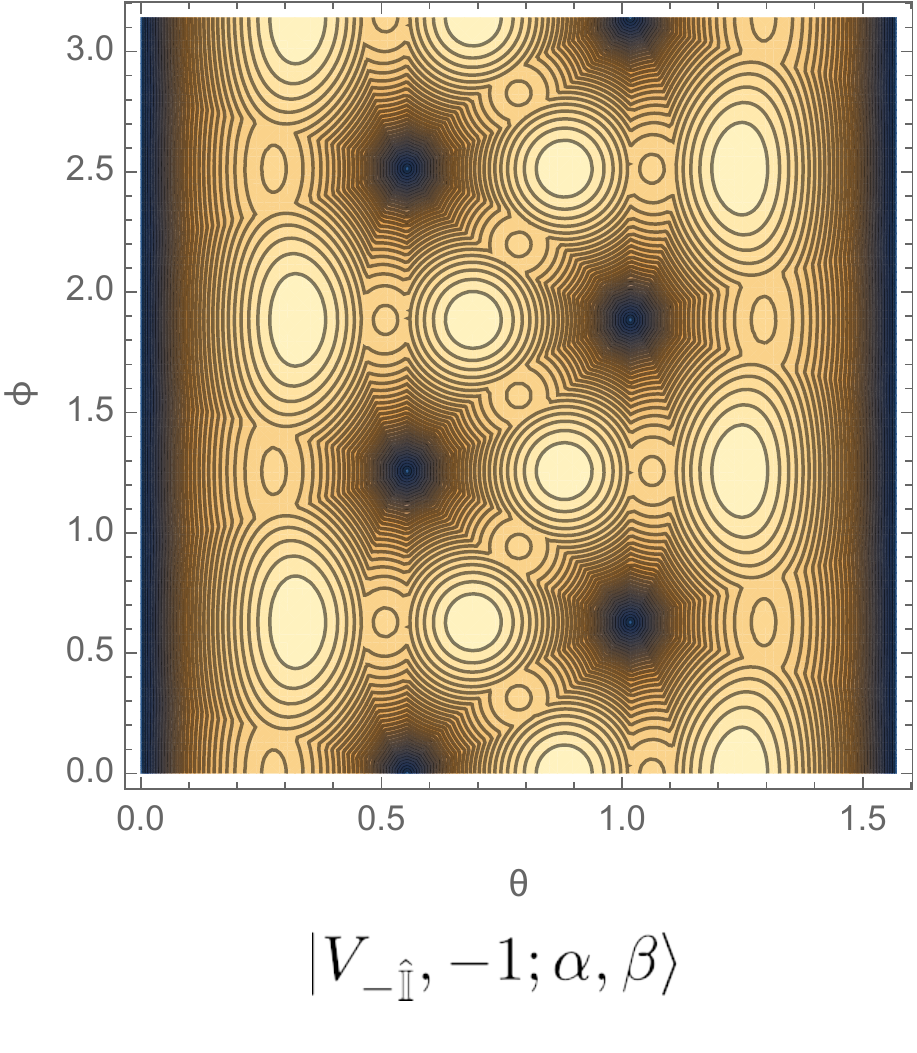}
\end{subfigure}
\begin{subfigure}[t]{0.49\textwidth}
    \includegraphics[align=c,width=.86\textwidth]{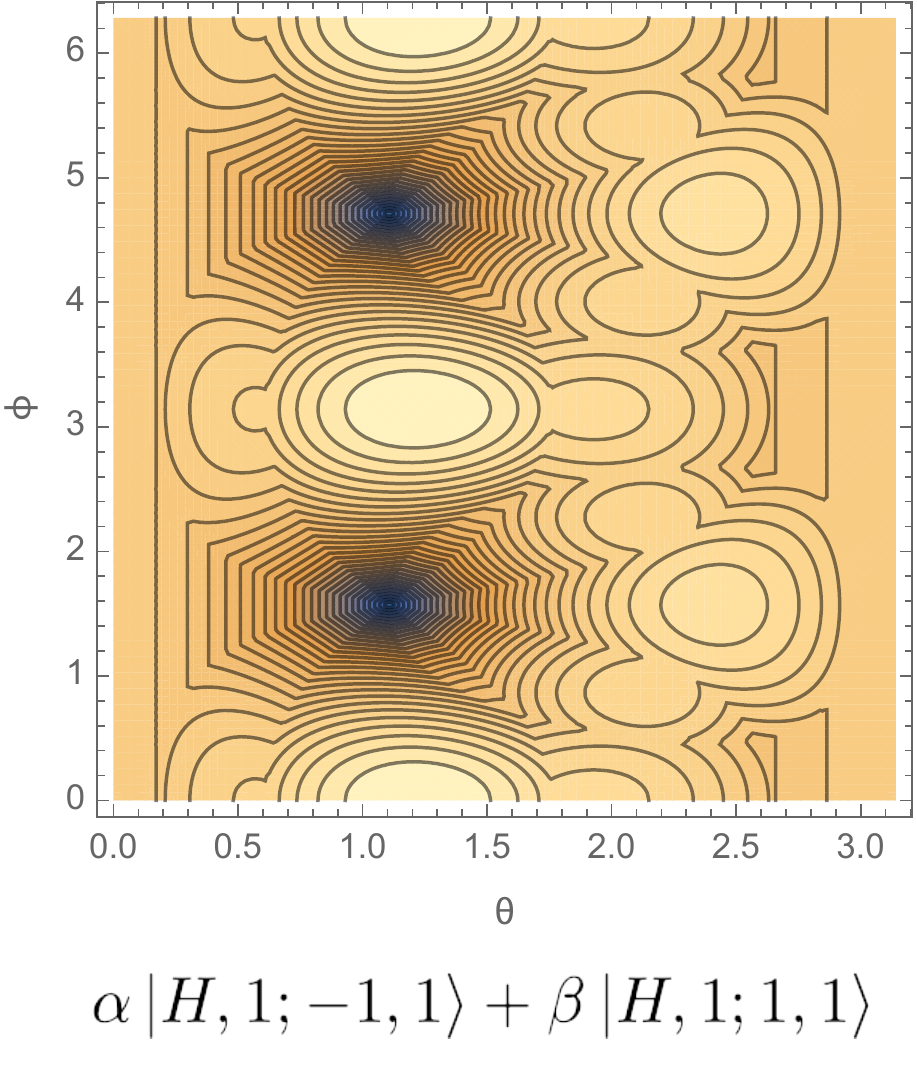}
\end{subfigure}  

\begin{subfigure}[t]{0.49\textwidth}
    \includegraphics[align=c,width=0.9\textwidth]{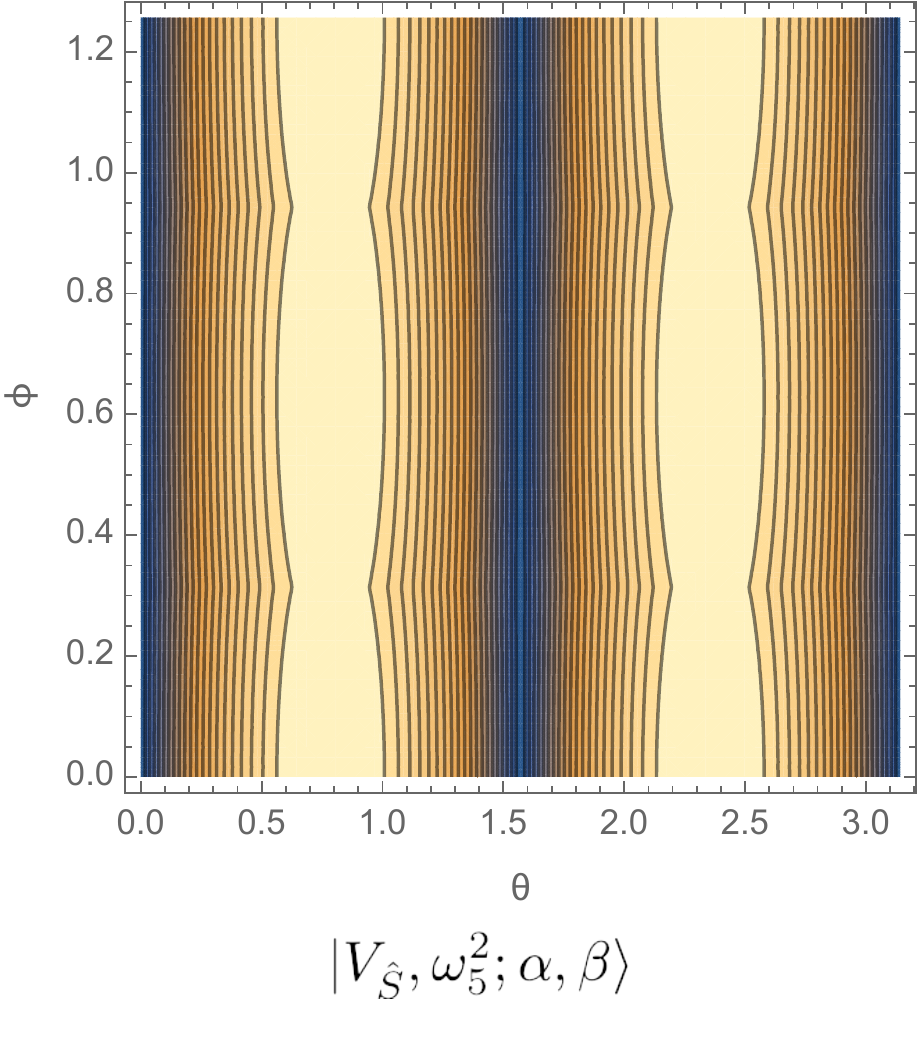}
\end{subfigure}
\includegraphics[align=c, width=0.055\textwidth]{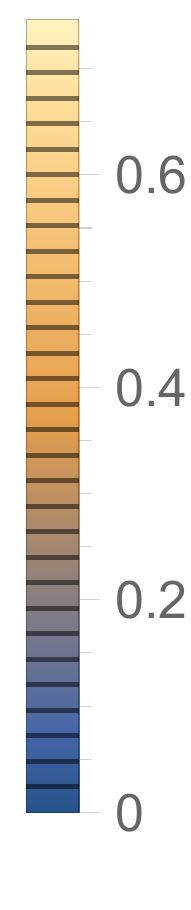}
\caption{The mana of each of the one-complex-parameter degenerate families of ququint Clifford eigenstates, with $\alpha=\cos \frac{\theta}{2}$ and $\beta=e^{i\phi} \sin \frac{\theta}{2}$. $\theta$ is the horizontal axis, and $\phi$ is the vertical axis. \label{degenerate5}}
\end{figure}

\section{Twirling Schemes}
\label{twirl-section}
A generic $p$-dimensional qudit density matrix is described by $p^2-1$ real parameters.  To avoid reduce the number of parameters describing noisy magic states, it is convenient to ``twirl" the undistilled resource state into a density matrix of a simpler form prior to distillation by randomly applying one or more Clifford operators that have the target magic state as an eigenvector, as described in \cite{MSD}.

For example, suppose $C$ is a single-qudit Clifford unitary such that $C^m=1$, with $p$ non-degenerate eigenvectors: $\ket{C_i}$ for $i=0, \ldots p-1$. Then twirling with respect to $C$ is a procedure that consists of randomly applying an element of $\langle C \rangle$, and results in the following:
\begin{equation}
    \rho \rightarrow \rho'=\frac{1}{m}\sum_{n=0}^{m-1} C^n\rho (C^\dagger)^n = \sum_{i=0}^{p-1} \alpha_i \ket{C_i}\bra{C_i}, \label{generic-twirl}
\end{equation}
where $\alpha_i = \bra{C_i} \rho  \ket{C_i}$. 

For qubits, via twirling, it was possible to restrict all forms of noise to depolarizing noise in \cite{MSD}. This is not usually true for qudits, where, instead, we generically expect twirling to reduce the number of parameters specifying the input qudit density matrices from $p^2-1$, to $p-1$, as demonstrated in equation \ref{generic-twirl}. If a Clifford unitary has degenerate eigenvectors, then, after twirling, we may be left with a density matrix whose description requires more than $p-1$ parameters. 

If we wish to distill a magic state which is an eigenstate of multiple Clifford operators, we may be able to reduce the number of parameters to fewer than $p-1$ by twirling multiple times. The diagram in Figure \ref{fig:venn} can be used to determine the inequivalent twirling schemes that may be applied before distilling any given magic state. 

The most general twirling procedure can be defined using any subgroup $\mathcal G$ of the Clifford group, and is the following:
\begin{equation}
    \rho \rightarrow \rho' = \frac{1}{|\mathcal G|}\sum_{G \in \mathcal G} G \rho G^\dagger.
\end{equation}
Define the \textit{stabilizing subgroup of the Clifford group} for state $\ket{M}$ as the set of elements of the Clifford group for which $\ket{M}$ is an eigenvector. (Here, we say $C$  ``stabilizes'' $\ket{M}$ when $C\ket{M}=\lambda\ket{M}$, for any $\lambda$.) The stabilizing subgroup of the Clifford group of $\ket{m}$ provides a natural scheme for twirling noisy $\ket{m}$ states.

\subsection{Qutrits}

We present the largest stabilizing subgroup of the Clifford group for each state in Table  \ref{qutrit-twirl-table}. We illustrate the twirling schemes for each state in detail below, which are also pictured in Figure \ref{qutrit-twirling-schemes}.

\begin{table}
\begin{center}
 \begin{tabular}{|l |c |c| c|}
 \hline
 \textbf{State} & \textbf{Generators} & \textbf{Order} & \textbf{Group} \\ [0.5ex]
 \hline\hline
 $\ket{S}$ & $\langle H,V_{\hat{S}} \rangle$ & 24 & $SL(2,\mathbb Z_3)$ \\
 \hline
 $\ket{H,1}$ & $\langle H\rangle$ & 4 & $C_4$  \\
 \hline
 $\ket{N_+}$ & $\langle N \rangle$ & 6 & $C_6$  \\
 \hline
 $\ket{XV_{\hat{S}}}$ & $\langle XV_{\hat{S}} \rangle$ & 3 & $C_3$ \\
 \hline
\end{tabular}
\end{center}
\caption{\textbf{The Stabilizing subgroup of the Clifford group for each qutrit magic state}. We list the largest subgroup of the Clifford group that stabilizes each magic state in the table above. The columns specify: the list of generators, order of the group and the name of the group. \label{qutrit-twirl-table}}
\end{table}

\begin{figure}[h]
    \centering
    \begin{subfigure}[t]{.3\textwidth}
    \includegraphics[width=.81\textwidth]{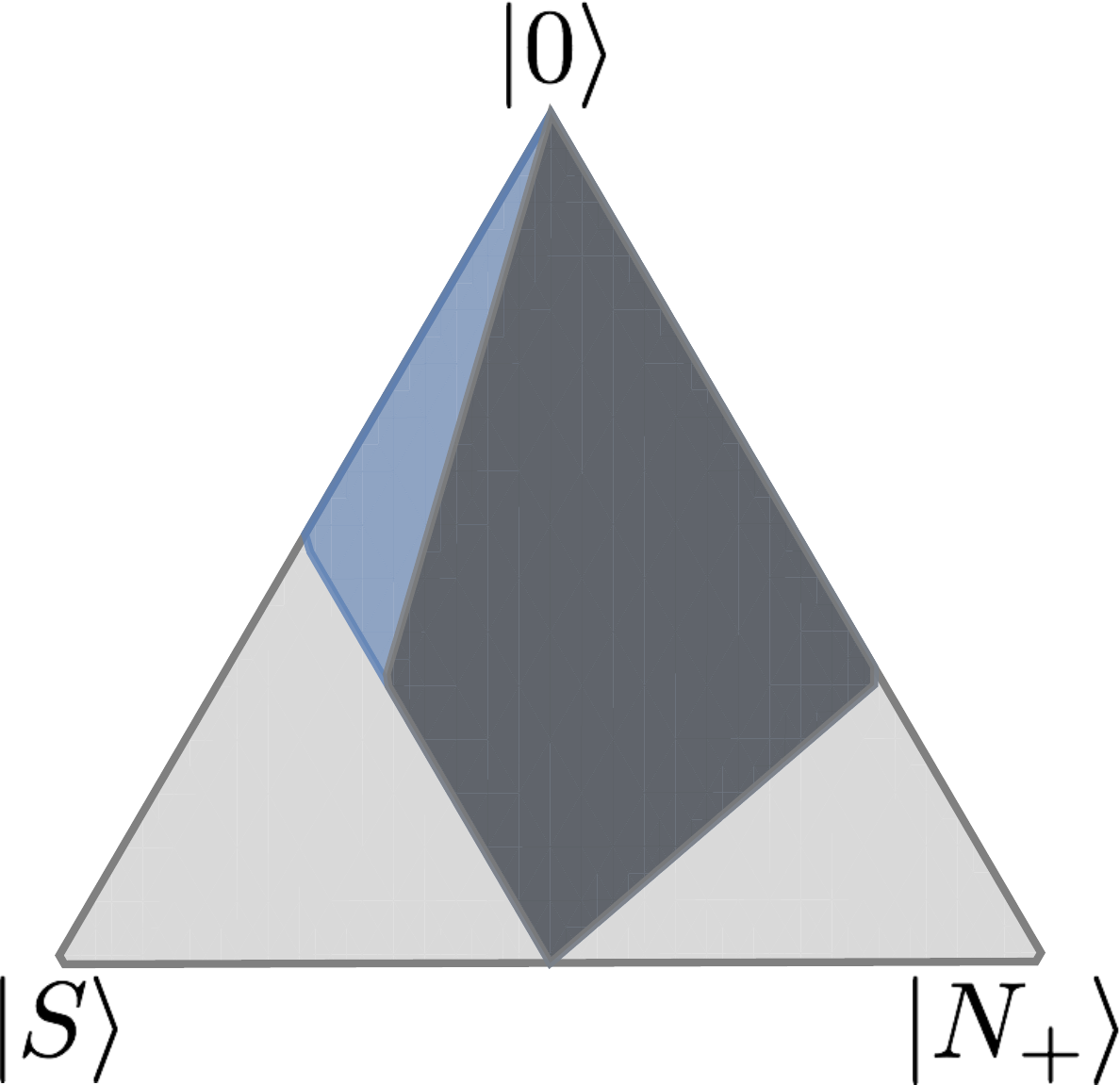}
    \end{subfigure}
    \begin{subfigure}[t]{.3\textwidth}
    \includegraphics[width=.9\textwidth]{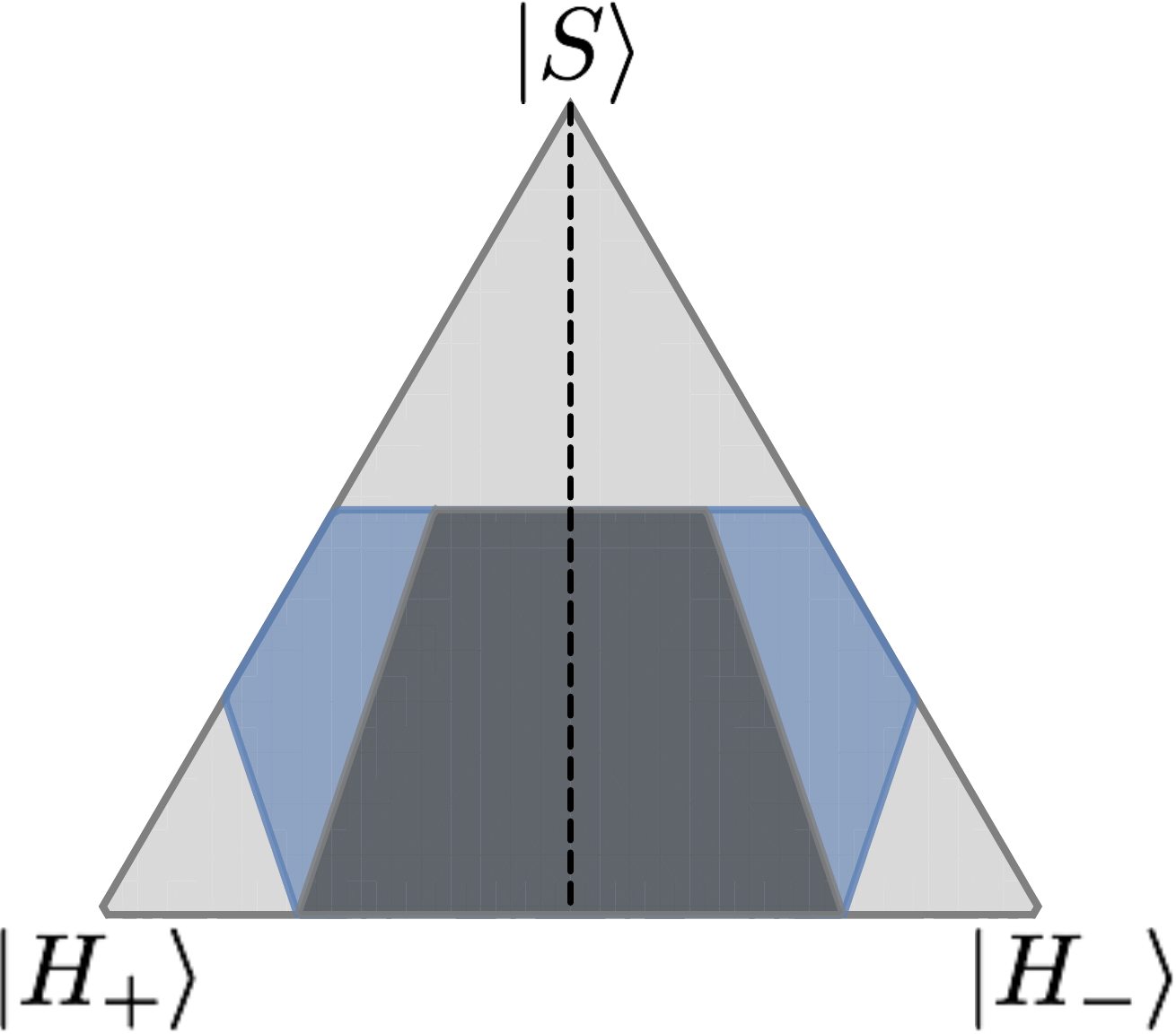}
    \end{subfigure}
    \begin{subfigure}[t]{.3\textwidth}
    \includegraphics[width=.94\textwidth]{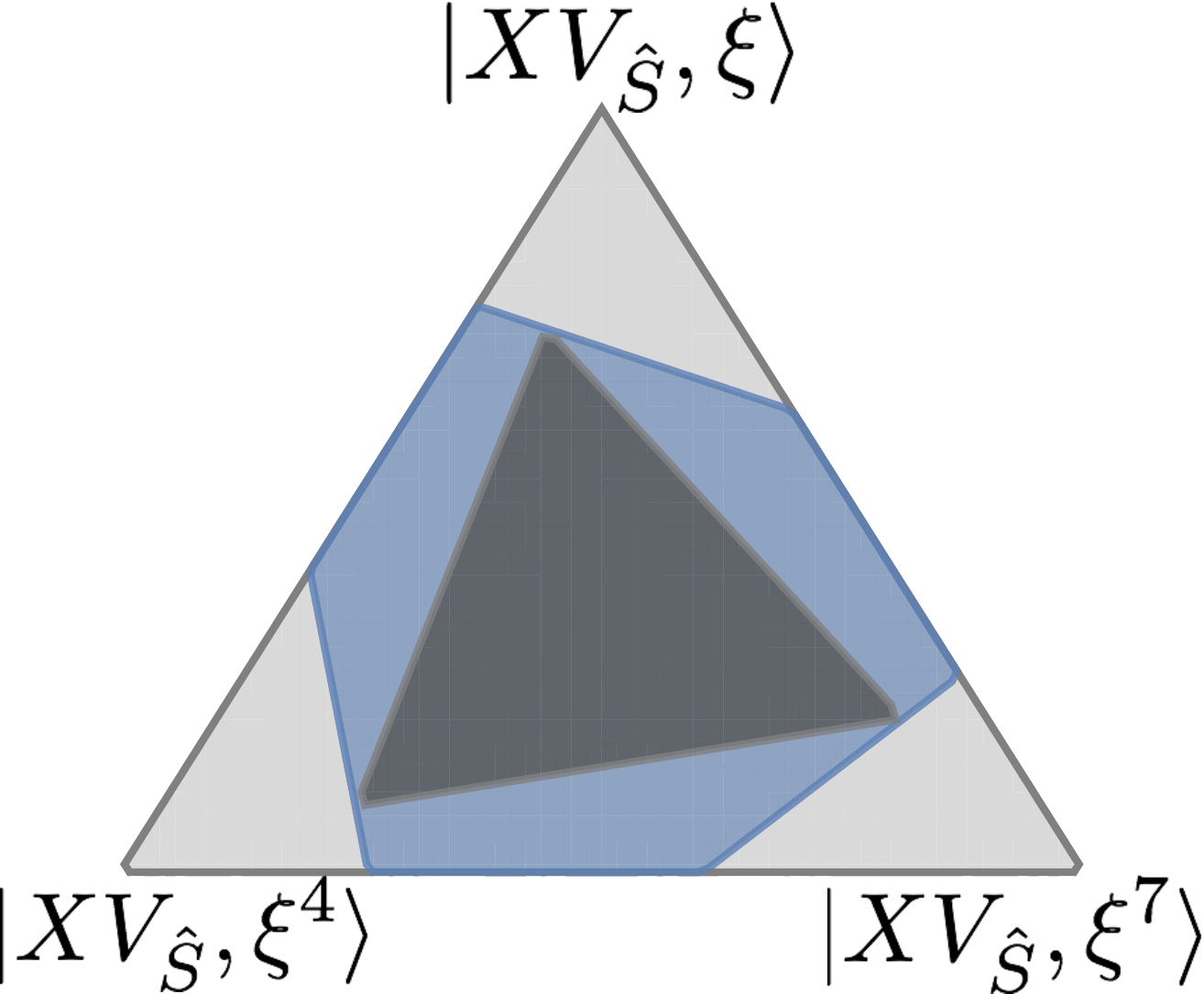}
    \end{subfigure}
    \caption{Noisy qutrit magic states can be twirled to lie in the planar slice of qutrit state space formed by the non-degenerate eigenvectors of each Clifford conjugacy class. Noisy strange states, can be further twirled to lie on the dashed line. The Wigner polytope is pictures as a blue region, which contains the Stabilizer polytope as a dark gray region. The centre of each equilateral triangle is the maximally mixed state.}
    \label{qutrit-twirling-schemes} \label{fig:XS-plane}
    \label{fig:H-plane}
    \label{fig:B-plane}
\end{figure}

\subsubsection{Twirling Schemes for Eigenstates of $H$}
To distill eigenstates of $H$, one can randomly apply the Clifford operator $H$ which will restrict our input state to the plane defined by $\ket{S}\bra{S}$, $\ket{H,1}\bra{H,1}$, and $\ket{H,{-1}}\bra{H,-1}$ given in Figure \ref{fig:H-plane}.
\begin{equation}
    \rho_H(\epsilon_1,\epsilon_2)=(1-\epsilon_1-\epsilon_2) \ket{S}\bra{S} + \epsilon_1 \ket{H,1}\bra{H,1}+ \epsilon_2\ket{H,{-1}}\bra{H,{-1}} . \label{2dH}
\end{equation}
 If we want to distill the $\ket{H,1}$ magic state, there is no further twirling possible.

If we wish to distill $\ket{S}$ magic states, we can further twirl to restrict density matrices to lie within a one-parameter family. To do this, we note that the states $\ket{H,1}$ and $\ket{H,-1}$ are interchanged by the symplectic rotation, $V_{-\hat{H}'}$, corresponding to the $SL(2,\mathbb Z_3)$ transformation $-\hat{H}'=\begin{pmatrix} 1 & 1 \\ 1 & 2 \end{pmatrix}$. Because it is a symplectic rotation, $V_{-\hat{H}'}$ preserves $\ket{S}$. We can further twirl by choosing randomly whether or not to apply this Clifford operator.

Input states are then restricted to the line which joins $\ket{S}\bra{S}$ to the maximally mixed state, shown as a dashed line in Figure \ref{fig:H-plane}, and can be parameterized as follows:
\begin{equation}
\begin{split}
    \rho_S(\epsilon/2,\epsilon/2) &=(1-\epsilon) \ket{S}\bra{S} + \epsilon \frac{\left( \ket{H,1}\bra{H,1}+\ket{H,{-1}}\bra{H,{-1}} \right)}{2}. \\
    & = (1-\delta) \ket{S}\bra{S} + \delta \frac{\mathbf{1}}{3}\label{1d},
    \end{split}
\end{equation}
where $\delta=\frac{3}{2} \epsilon$ is the depolarizing noise rate.

The final result, equation \ref{1d}, can also be obtained by simply applying a random symplectic rotation $V_{\hat{F}}$:
\begin{equation}
    \rho \rightarrow \rho' = \sum_{\hat{F} \in SL(2,\mathbb Z_3)} V_{\hat{F}} \rho V_{\hat{F}}^\dagger.
\end{equation} 

For all other magic states, it appears that twirling can only restrict the input state to lie in a plane. So the scheme given above is unique to $\ket{S}$.

\subsubsection{Twirling Schemes for Eigenstates of $N$}
An alternative twirling scheme is to randomly apply the $N$ operator, so that states are restricted to the plane defined by its eigenvectors given in Figure \ref{fig:B-plane}. This can be used to distill $\ket{S}$ or $\ket{N_+}$ states. A density matrix in this plane can be expressed as:
\begin{equation}
    \rho_N(\epsilon_1, \epsilon_2) = (1-\epsilon_1-\epsilon_2) \ket{N_+}\bra{N_+} + \epsilon_1 \ket{0}\bra{0}+\epsilon_2 \ket{S}\bra{S}.
\end{equation}

For stabilizer codes used with this twirling scheme, $N$ should be a transversal operator.

\subsubsection{Twirling Schemes for Eigenstates of $XV_{\hat{S}}$}
A third twirling scheme is to randomly apply the $XV_{\hat{S}}$ operator, so that states are restricted to the plane defined by its eigenvectors given in Figure \ref{fig:XS-plane}, which is useful for distilling $\ket{XV_{\hat{S}} }$ states. A density matrix in this plane can be expressed as,
\begin{equation}
    \rho_{XV_{\hat{S}}}(\epsilon_1, \epsilon_2) = (1-\epsilon_1 - \epsilon_2) \ket{XV_{\hat{S}} }\bra{XV_{\hat{S}} } + \epsilon_1 \ket{XV_{\hat{S}}' }\bra{XV_{\hat{S}}' }
    + \epsilon_2 \ket{XV_{\hat{S}}'' }\bra{XV_{\hat{S}}'' }.
\end{equation}
where $\ket{XV_{\hat{S}}' }$ and $\ket{XV_{\hat{S}}'' }$ are the other eigenstates of $XV_{\hat{S}}$.

For stabilizer codes used with this twirling scheme, $XV_{\hat{S}}$ should be a transversal operator.

\subsubsection{Twirling Schemes for Degenerate Families of States}
For the degenerate families of states $\ket{V_{\hat{S}},\omega_3^2}$, we can apply $V_{\hat{S}}$ a random number of times. The resulting space of density matrices will be 4-dimensional: 3-real parameters for the ``Bloch sphere'' of degenerate $\ket{V_{\hat{S}},\omega_3^2}$ states, and additional parameter for the state $\ket{0}$. Similar comments apply for distilling, the states $\ket{V_{-\hat{\mathbb{I}}},1}$.

Explicitly, after randomly applying $V_{\hat{S}}$, any density matrix can be put in the form,
\begin{equation}
\rho(x,y,z,\epsilon) = (1-\epsilon)\frac{1}{2}\left( \ket{1}\bra{1}+\ket{2}\bra{2} + x \Sigma_1 + y \Sigma_2 + z \Sigma_3 \right)+\epsilon \ket{0}\bra{0}
\end{equation}
where
\begin{equation}
\Sigma_1 = \ket{1}\bra{2}+\ket{2}\bra{1}, ~ \Sigma_2 = -i \ket{1}\bra{2}+i\ket{2}\bra{1}, ~ \Sigma_3 = \ket{1}\bra{1} - \ket{2}\bra{2}.
\end{equation}

After randomly applying $V_{-\hat{\mathbb{I}}}$, any density matrix can be put in the form,
\begin{equation}
\tilde{\rho}(x,y,z,\epsilon) = (1-\epsilon)\frac{1}{2}\left( \ket{0}\bra{0}+\ket{N_+}\bra{N_+} + x \tilde{\Sigma}_1 + y \tilde{\Sigma}_2 + z \tilde{\Sigma}_3 \right)+\epsilon \ket{S}\bra{S}
\end{equation}
where
\begin{equation}
\tilde{\Sigma}_1 = \ket{0}\bra{N_+}+\ket{N_+}\bra{0}, ~ \tilde{\Sigma}_2 = -i \ket{0}\bra{N_+}+i\ket{N_+}\bra{0}, ~ \tilde{\Sigma}_3 = \ket{0}\bra{0} - \ket{N_+}\bra{N_+}.
\end{equation}

\subsection{Ququints}
We present the largest stabilizing subgroup of the Clifford group for each ququint magic state in Table  \ref{ququint-twirl-table}. These translate into twirling schemes in a straightforward way. Most of the schemes result in spaces with $4$ or more parameters. However, the two cases of $\ket{H,-1}$ and $\ket{B',-1}$, give rise to smaller spaces after twirling, as we discuss below.

\begin{table}
\begin{center}
 \begin{tabular}{|l |c |c| c|}
 \hline
 \textbf{State} & \textbf{Generators} & \textbf{Order} & \textbf{Group} \\ [0.5ex]
 \hline\hline
 $\ket{H_i}$ & $\langle H \rangle$ & 4 & $C_4$ \\
 \hline
 $\ket{H,{-1}}$ & $\langle H,H' \rangle$ & 8 & Quaternion \\
 \hline
 $\ket{{B,-1}}$ & $\langle B,H' \rangle$ & 12 & Dicyclic$_{3}$ \\
 \hline
 $\ket{B,-{\omega_3^2}}$ & $\langle B \rangle$ & 6 & $C_6$ \\
 \hline
 $\ket{B,+{\omega_3^2}}$ & $\langle B \rangle$ & 6 & $C_6$ \\
 \hline
 $\ket{A,-\omega_5^2} $ & $\langle A \rangle$ & 10 & $C_{10}$ \\
 \hline
 $\ket{A,\omega_5^2}$ & $\langle A \rangle$ & 10 & $C_{10}$ \\
 \hline
 $\ket{C}$ & $\langle XV_{\hat{S}} \rangle$ & 5 & $C_{5}$ \\
 \hline
\end{tabular}
\end{center}
\caption{\textbf{The stabilizing subgroup of the Clifford group for each non-degenerate ququint magic state}. We list the largest subgroup of the Clifford group that stabilizes each ququint magic state in the table above. The columns specify: the list of generators, order of the group and the name of the group. \label{ququint-twirl-table}}
\end{table}

\subsubsection{Twirling Scheme for the State $\ket{H,-1}$}
We first randomly apply $H$ to restrict ourselves to mixtures of $\ket{H,\pm i}$, $\ket{H,-1}$ and $\ket{H,1; \alpha, \beta}$. This is a $6$-parameter space.

The operator
\begin{equation}
H' = V_{\begin{pmatrix} 0 & 2 \\ 2 & 0 \end{pmatrix}},
\end{equation}
is in the same conjugacy class as $[[H]]$. Together $H$ and $H'$ generate a non-abelian group of $8$ elements, isomorphic to the quaterion group:
\begin{equation}
\text{Quaternion} = \left\langle H, ~H' ~|~ H^4=1,~{H'}^2=H^2,~HH'H=H' \right\rangle.
\end{equation}

The state $\ket{H,-1}$ is also an eigenstate of $H'$ with eigenvalue $-1$. $H'$ acts on the other eigenstates of $H$ as follows:
\begin{equation}
H'\ket{H, +i}=\ket{H, -i}, ~  H'\ket{H,1;1,1} = \ket{H,1;1,1}, ~ H'\ket{H,1;1,-1} = -\ket{H,1;1,-1}.
\end{equation}

Thus by randomly applying $H$ then $H'$, we obtain the three-parameter family of density matrices,
\begin{dmath}
\rho = (1-\epsilon_1-\epsilon_2-\epsilon_3)\ket{H,-1}\bra{H,-1} + \epsilon_1 \ket{H,1;1,1}\bra{H,1;1,1}+ \epsilon_2 \ket{H,1;1,-1}\bra{H,1;1,-1}+\epsilon_3 \rho_i , \label{hm5}
\end{dmath}
where,
\begin{dmath}
\rho_i = \frac{1}{2}\left( \ket{H, +i}\bra{H, +i} + \ket{H,-i}\bra{H,i} \right).
\end{dmath}
This region is shown in Figure \ref{hm1Wigner5}.
\begin{figure}
\centering
    \includegraphics[width=.5\textwidth]{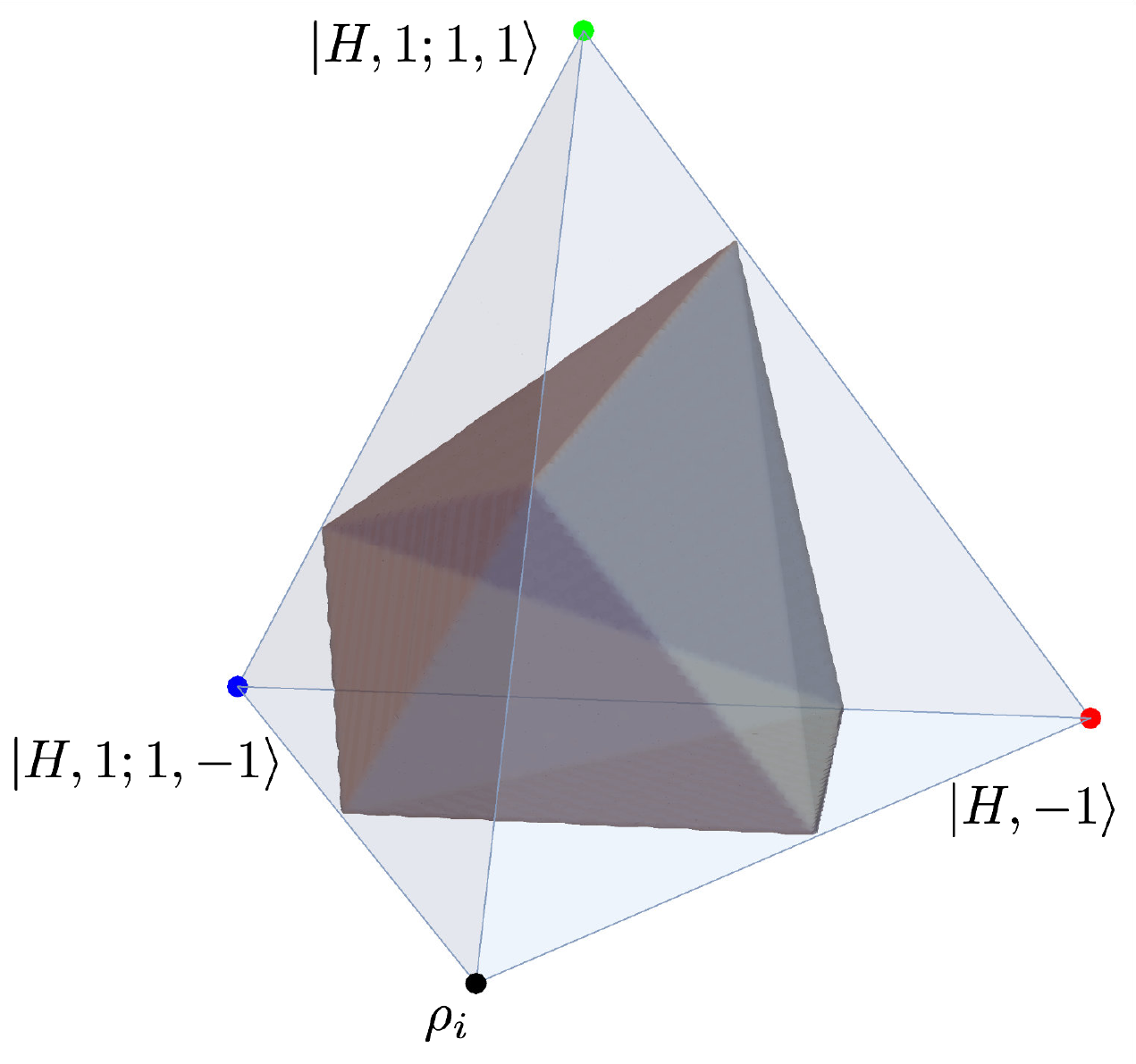}
    \caption{Noisy ququint $\ket{H,-1}$ states can be twirled to lie in the 3-dimensional convex mixture of states defined by Equation \eqref{hm5}, depicted as a light-blue tetrahedron. The gray polytope inside is the Wigner polytope. The stabilizer polytope is not pictured.}
    \label{hm1Wigner5}
\end{figure}

\subsubsection{Twirling Scheme for the  State $\ket{B',-1}$}
The state $\ket{B,-1}$ is an eigenvector of both $B$ and $H'$. Together $B$ and $H'$ generate a non-abelian group of $12$ elements known as the dicyclic group of order $12$, sometimes written as Dicyclic${}_{3}$. It can be presented as:
\begin{equation}
\text{Dicyclic}{}_3 = \left\langle B, ~H' ~|~ B^6 = 1,~H'{}^2 = B^3, ~ H'{}^{-1} B H' = B^{-1} \right\rangle.
\end{equation}

$H'$ acts on the other eigenstates of $B$ $\ket{B,\pm \omega_3^2}$ and $\ket{B,\pm \omega_3^3}$, as follows:
\begin{equation}
H' \ket{B, +\omega^3}=\ket{B,+\omega^2}, ~ H' \ket{B, -\omega^3}=\ket{B,-\omega^2}.
\end{equation}

By randomly applying $B$, we are left with the four parameter family of convex combinations of $\ket{B,-1}$, $\ket{B,\pm \omega_3^2}$ and $\ket{B,\pm \omega_3^3}$. We then apply $H'$ randomly to restrict our space to the two-parameter family of density matrices given by:
\begin{equation}
\rho_{B,-1}(\epsilon_+,\epsilon_-) = (1-\epsilon_+-\epsilon_-)\ket{B,-1}\bra{B,-1} + \epsilon_+ \frac{1}{2}\rho_+ +\epsilon_- \rho_-. \label{Bm15}
\end{equation}
Here, $\rho_\pm =\frac{1}{2}\left( \ket{B, \pm \omega_3^2}\bra{B, \pm \omega_3^2} + \ket{B, \pm \omega_3^3}\bra{B, \pm\omega_3^3} \right)$.
This is shown in Figure \ref{B1ququint}.

\begin{figure}
\centering
    \includegraphics[width=.3\textwidth]{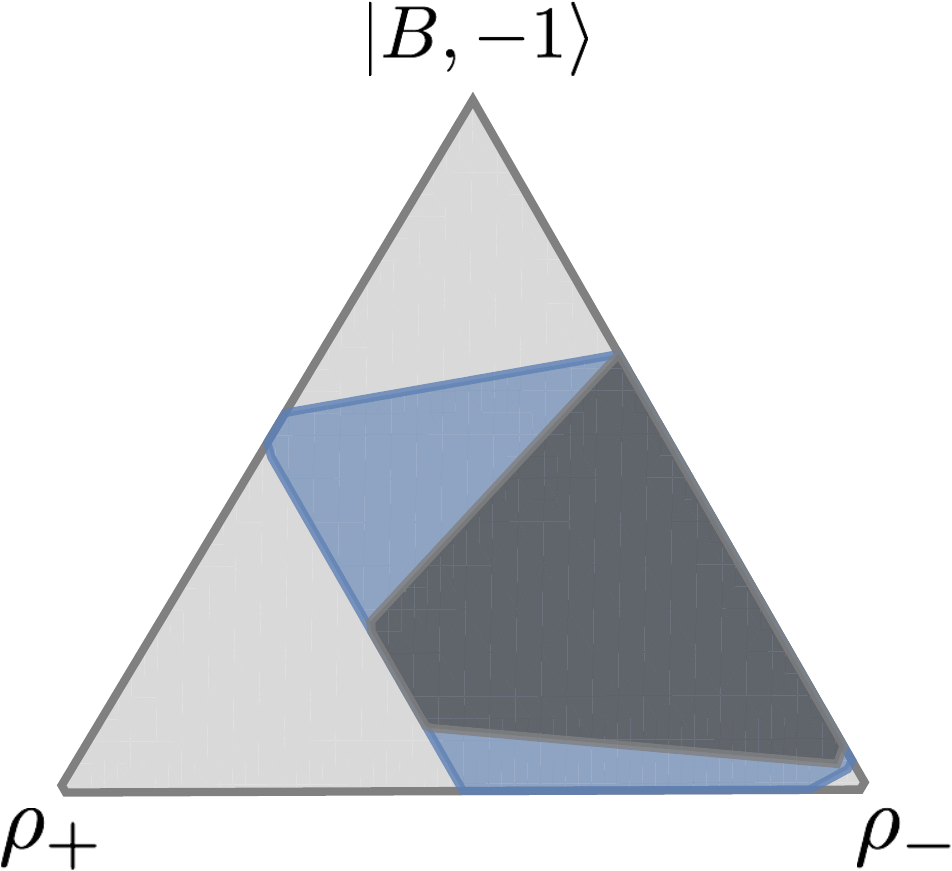}
    \caption{Noisy ququint $\ket{B,-1}$ states can be twirled to lie in a the convex region defined by Equation \eqref{Bm15}, depicted as a light gray triangle. The dark blue region inside the triangle is the Wigner polytope. and the dark gray region is the stabilizer polytope.}
    \label{B1ququint}
\end{figure}

\section{Uniqueness of the Qutrit Strange State}
\label{qutrit-strange-state}
Of the qutrit and ququint magic states presented in the previous sections, the qutrit strange state $\ket{S}$ stands out, as it has several interesting properties. The qutrit strange state is identified as the most magic qutrit state, by virtue of its maximal mana and thauma \cite{Veitch_2014,wang2018efficiently}. The qutrit strange state is also, in principle, maximally robust to polarizing noise \cite{Howard_2013}. From the point of view of Clifford symmetries, the strange state is also distinguished as simultaneous eigenvector of all symplectic rotations. As such, it has a particularly simple discrete Wigner function, from which we see that it lies directly above the centre of a single facet of the Wigner polytope.

Does an analogue of the qutrit strange state exist for qudits of higher $p$ that simultaneously possesses all these properties? In any dimension, one can determine which states maximize the mana, and also determine which states that are maximally robust to depolarizing noise; however, we saw that, for $p=5$, the states that maximize the mana are not equivalent to the states that are maximally robust to depolarizing noise. 

We can also ask, does there exist a simultaneous eigenvector of all symplectic rotations for qudits of any other odd prime dimension $d$?
\begin{theorem}
There is no pure state that is a simultaneous eigenvector of all symplectic rotation, for qudits of odd prime dimension $d>3$. \label{thm1}
\end{theorem}
\begin{proof}
A (non-stabilizer) eigenvector of all symplectic rotations must be an eigenvector of the generators of symplectic rotations $H$ and $V_{\hat{S}}$. It is easy to see that there is no simultaneous eigenvector of $V_{\hat{S}}$ and $H$ for odd prime $p>3$.

$V_{\hat{S}}$ is diagonal in the computational basis. Its eigenvalues are $\lambda_i=\omega_p^{2^{-1} q_i}$ where $q$ is any quadratic residue mod $p$, i.e., any element of $\mathbb Z_p$ such that $q \equiv t^2 \mod p$ for some $t$. There $(p+1)/2$ such quadratic residues, including $0$. The eigenvector corresponding to $q=0$ is $\ket{0}$ which is non-degenerate. The remaining $(p-1)/2$ quadratic residues each have two-dimensional degenerate eigenspaces of the form:
\begin{equation}
    \ket{V_{\hat{S}}, \omega_p^{2^{-1} q}= \alpha,\beta}=\alpha \ket{t}+\beta\ket{-t}
\end{equation}
where $t$ and $-t$ are the two solutions to the equation $t^2 \equiv q \mod p$.

If a simultaneous eigenvector of $H$ and $V_{\hat{S}}$ exists, then $\ket{V_{\hat{S}}, \omega_p^{2^{-1} q};\alpha,\beta}$ must be an eigenvector of $H$ for some $q\neq 0$ and some choice of $\alpha$ and $\beta$. Let us see that this cannot be the case.

\begin{equation}
    H\ket{V_{\hat{S}}, \omega_p^{2^{-1} q};\alpha,\beta} = \sum_j \left(\alpha \omega_p^{jt}+\beta \omega_p^{-jt}\right) \ket{j}. \label{h-e}
\end{equation}
For $\ket{V_{\hat{S}}, \omega_p^{2^{-1} q};\alpha,\beta}$ to be an eigenvector of $H$ the coefficient of $\ket{0}$ must vanish in the above expression. This means that $\alpha=-\beta$. If $d>3$, then there exists another $k \neq 0,~\pm t$, such that the coefficient of $\ket{k}$ must also vanish in equation \eqref{h-e}. This implies that $\alpha \left(\omega_p^{kt}-\omega_p^{-kt}\right)=2i \alpha \sin(2\pi kt/p)=0$, which means $\alpha=0$ and no simultaneous eigenvector of $H$ and $V_{\hat{S}}$ exists.
\end{proof}
Our interpretation of this result is that the qutrit strange state is distinguished not only as both the most symmetric and magic qutrit state, but also the most symmetric of all qudit magic states.

Note that there are mixed states, other than the maximally mixed state, which are preserved by all symplectic rotations. Let $\rho_S$ be an equal mixture of all stabilizer states without support on the phase space point $A_{0,0}$. (This excludes $1$ basis vector from each of the $d+1$ mutually unbiased bases.) $\rho_S$ has discrete Wigner function given by:
\begin{equation}
W_{(u,v)}(\rho) = \begin{cases} 0 & (u,v)=(0,0) \\ \frac{1}{p^2-1} & (u,v) \neq (0,0)\end{cases}. \label{a00}
\end{equation}
Clearly, this Wigner function is preserved by all symplectic rotations.

Since the space of pure qudit density matrices is $2p-2$-dimensional, and the space of all qutrit density matrices is $p^2-1$-dimensional, most points at the boundary of qudit state space are not pure density matrices. Hence, there is no reason to expect the state directly above the centre of a face of the Wigner polytope to be a pure state, and the above argument shows that, for $p>3$, it is not a pure state.

The absence of a state directly above a facet of the Wigner polytope does not mean that there are no states which maximally violate the contextuality inequality of \cite{nature}. There is always a conjugacy class $[[V_{-\hat{\mathbb{I}}}]]$, corresponding to $\begin{pmatrix} -1 & 0 \\ 0 & -1 \end{pmatrix} \in SL(2,\mathbb Z_p)$. $V_{-\hat{\mathbb{I}}}$ can be written as
\begin{equation}
    V_{-\hat{\mathbb{I}}}= \sum_{k=0}^{p-1} \ket{-k}\bra{k}.
\end{equation}
Its eigenvalues are $1$, which has degeneracy $(p+1)/2$, and $-1$, which has degeneracy $(p-1)/2$. The corresponding eigenvectors were studied in \cite{noiseQudit}:
\begin{eqnarray*}
\ket{V_{1}, 1; \alpha_i} & = & \alpha_0 \ket{0} + \sum_{j=1}^{(p-1)/2} \frac{\alpha_j}{\sqrt{2}} (\ket{j}+\ket{-j}) \\
\ket{V_{1}, -1; \beta_i} & = &  \sum_{j=1}^{(p-1)/2} \frac{\beta_j}{\sqrt{2}} (\ket{j}-\ket{-j}).
\end{eqnarray*}
From the results of \cite{nature}, the minimum entry of the (normalized) discrete Wigner function is $-1/p$. Since $A_{00}=V_{-\hat{\mathbb{I}}}$, this shows that there are always Clifford eigenstates that attain this minimum value, i.e., that maximally violate the contextuality inequality. As emphasized in \cite{noiseQudit}, these states have the potential to be distilled with the largest threshold to depolarizing noise.

\section*{Acknowledgements}

SP would like to thank Rev. Prof. PS Satsangi for guidance. This research is supported in part by a DST INSPIRE Faculty award, DST-SERB Early Career Research Award (ECR/2017/001023) and MATRICS grant (MTR/2018/001077).

\section*{Appendix: Maximizing Mana for Qutrits}
In this section, we determine the qutrit states which locally and globally maximize the mana. While this was done numerically in \cite{Veitch_2012}, here we proceed analytically. This provides an alternative characterization of qutrit magic states as local maxima of the mana.

The most general qutrit state can be written as
\begin{equation}
    \ket{f(\theta, \phi, \psi_1, \psi_2)} = \cos \theta \ket{0}+\sin \theta \left( e^{i\psi_1} \cos \phi  \ket{1} + e^{i\psi_2} \sin \phi  \ket{2}\right).
\end{equation}
The range of these variables is $\theta \in [0,\pi/2]$, $\phi \in [0,\pi/2]$, $\psi_i \in [0,2\pi]$. We define $\psi_-=\psi_1-\psi_2$. 

The discrete Wigner function of $\ket{f(\theta, \phi, \psi_1, \psi_2)}$ is
\begin{eqnarray}
W_{(0,0)} & = & \frac{1}{3} \left(\cos ^2\theta+\sin ^2\theta \sin (2 \phi ) \cos \psi_- \right) \\
W_{(0,1)} & = & \frac{1}{3} \left(\cos^2 \theta+ \sin ^2\theta \sin (2 \phi ) \sin (\psi_- - \pi/6)\right) \\
W_{(0,2)} & = & \frac{1}{3} \left(\cos^2 \theta- \sin ^2\theta \sin (2 \phi ) \sin (\psi_- + \pi/6) \right) \\
W_{(1,0)} & = & \frac{1}{3} \left(\sin ^2\theta \cos ^2\phi+\sin (2 \theta )  \sin \phi \cos \psi_2\right) \\
W_{(1,1)} & = & \frac{1}{3} \left( \sin ^2\theta \cos ^2\phi+\sin (2 \theta ) \sin \phi  \sin (\psi_2 - \pi/6)\right) \\
W_{(1,2)} & = & \frac{1}{3} \left( \sin ^2\theta \cos ^2\phi-\sin (2 \theta ) \sin \phi  \sin (\psi_2 + \pi/6)\right) \\
W_{(2,0)} & = & \frac{1}{3} \left(\sin ^2\theta \sin ^2\phi+\sin (2 \theta )  \cos \phi \cos \psi_1 \right) \\
W_{(2,1)} & = & \frac{1}{3} \left(\sin ^2\theta \sin ^2\phi-\sin (2 \theta ) \cos \phi  \sin (\psi_1+\pi/6)\right) \\
W_{(2,2)} & = & \frac{1}{3} \left(\sin ^2\theta \sin ^2\phi+\sin (2 \theta ) \cos \phi  \sin (\psi_1-\pi/6)\right).
\end{eqnarray}
Without loss of generality, we can also restrict the range of $\psi_i$ to be $[0,2\pi/3]$, which implies that the range of $\psi_-$ is $[-2\pi/3,2\pi/3]$. (States with other values of $\psi_i$ can be obtained from states in the above range, by applying the Clifford operators $Z$ and $S$.)
 
Maximizing the mana is equivalent to maximizing the sum-negativity of the Wigner function. Recall the sum-negativity is the absolute value of the sum of negative entries in the discrete Wigner function. Each of the above entries are negative when the following conditions are met:
\begin{eqnarray}
W_{(0,0)}<0 &~~~ \rightarrow ~~~& \sin (2\phi) > \cot^2 \theta,~3\pi/2>|\psi_-|>\pi/2 \\
W_{(0,1)}<0 &~~~ \rightarrow ~~~& \sin (2\phi) > \cot^2 \theta,~\pi<\psi_- - \pi/6<2\pi \\
W_{(0,2)}<0 &~~~ \rightarrow ~~~& \sin (2\phi) > \cot^2 \theta,~0<\psi_-+\pi/6<\pi \\
W_{(1,0)}<0 &~~~ \rightarrow ~~~& 2\tan \phi \sec \phi > \tan \theta,~3\pi/2>|\psi_2|>\pi/2 \\
W_{(1,1)}<0 &~~~ \rightarrow ~~~& 2\tan \phi \sec \phi > \tan \theta,~\pi<\psi_2 - \pi/6<2\pi \\
W_{(1,2)}<0 &~~~ \rightarrow ~~~& 2\tan \phi \sec \phi > \tan \theta,~0<\psi_2+\pi/6<\pi \\
W_{(2,0)}<0 &~~~ \rightarrow ~~~& 2\cot \phi \csc \phi > \tan \theta,~3\pi/2>|\psi_1|>\pi/2 \\
W_{(2,1)}<0 &~~~ \rightarrow ~~~& 2\cot \phi \csc \phi > \tan \theta,~0<\psi_1+\pi/6<\pi \\
W_{(2,2)}<0 &~~~ \rightarrow ~~~& 2\cot \phi \csc \phi > \tan \theta,~\pi<\psi_1 - \pi/6<2\pi \\
\end{eqnarray}

We now proceed by a case analysis based on the number of negative entries in this discrete Wigner function.

\begin{itemize}
\item[Case 1:] For the state to have non-zero mana, it must have at least one negative entry its discrete Wigner function.  Wherever it is located, it can be shifted to $(0,0)$ by a Heisenberg-Weyl displacement. To find the state with maximum mana, we must minimize 
\begin{equation}
    W_{(0,0)}=\frac{1}{3}\left(\cos^2 \theta + \sin^2 \theta \sin (2\phi) \cos \psi_-\right),
\end{equation} subject to the constraint that all other entries $W_{(i,j)}\geq 0$. We find $W_{(0,0)}$ is minimized when $\cos \theta=0$, $\sin (2\phi)=1$ and $\cos \psi_-=-1$; this corresponds to a sum-negativity of $1/3$. This state is the strange state. 

\item[Case 2:] If the Wigner function has two negative entries, they can be moved to $(0,0)$ and $(0,1)$ by first applying a Heisenberg-Weyl translation, then a symplectic rotations. To find the state of this form with maximum sum-negativity, we must minimize 
\begin{equation}
    W_{(0,0)}+W_{(0,1)}=\frac{1}{3}\left(2\cos^2 \theta + \sin^2 \theta \sin (2\phi) \left(\cos \psi_-+\sin (\psi_--\pi/6)\right)\right),
\end{equation}

subject to the constraint that both $W_{(0,0)}$ and $W_{(0,1)}$ are negative, and all other entries of the Wigner function are non-negative. We find $W_{(0,0)}$ is minimized when $\cos \theta=0$, $\sin (2\phi)=1$ and $\psi_-=4\pi/3$; this also corresponds to a sum-negativity of $1/3$. This state is Clifford-equivalent to the Norell state.

\item[Case 3:] If the discrete Wigner function has three negative entries, the first two negative entries can be taken to be located at $(0,0)$ and $(0,1)$. There are two Clifford-inequivalent possibilities for the third negative entry: $(0,2)$ and $(1,0)$. It is easy to see that it is impossible for $W_{(0,0)}$, $W_{(0,1)}$ and $W_{(0,2)}$ to all simultaneously be negative; so we take the third negative entry to be $W_{(1,0)}$. We now minimize 
\begin{equation}\begin{split}
    &W_{0,0}+W_{0,1}+W_{1,0} =\\ &\frac{1}{3}\left(2\cos^2 \theta + \sin^2 \theta \sin (2\phi) \left(\cos \psi_-+\sin (\psi_--\pi/6)\right) + \sin ^2\theta \cos ^2\phi+\sin (2 \theta )  \sin \phi \cos \psi_2\right),
    \end{split}
\end{equation} 
subject to the constraint that $W_{(0,0)}$, $W_{(0,1)}$ and $W_{(1,0)}$ are all negative, and all other entries of the discrete Wigner function are non-negative. 
We find the sum-negativity comes out to be $\frac{1}{3} \left(-1+2 \cos \left(\frac{\pi}{9}\right)\right)$ and the state is Clifford equivalent to the equatorial magic state $\ket{XV_{\hat{S}}}$.

\item[Case 4:] We now consider the case where 4 entries of the discrete Wigner function are negative. By the discussion above, we take the first three negative entries to be at $(0,0)$, $(0,1)$ and $(1,0)$. Notice that it is impossible for the Wigner function to be negative at three collinear points. Therefore the only possibilities for the fourth point are $(1,1)$, $(2,1)$ and $(1,2)$. It turns out that all these choices are related to each other by a Clifford transformation, so we can assume the fourth point is located at $(1,1)$. We thus have to minimize:
\begin{equation}
\begin{split}
   & W_{(0,0)}+W_{(0,1)}+W_{(1,0)}+W_{(1,1)}
    =\\ &\frac{1}{3}\left(2\cos^2 \theta + \sin^2 \theta \sin (2\phi) \left(\cos \psi_-+\sin (\psi_--\pi/6)\right) + \sin ^2\theta \cos ^2\phi + \sin (2 \theta )  \sin \phi \left(\cos \psi_2+\sin (\psi_2-\pi/6)\right)\right),
    \end{split}
\end{equation}
subject to the constraint that $W_{(0,0)}$, $W_{(0,1)}$, $W_{(1,0)}$ and $W_{(1,1)}$ are negative, and the remaining entries of the discrete Wigner function are non-negative. We find the maximal sum-negativity is $\frac{1}{3} \left(\sqrt{3}-1\right)$, and the state is Clifford equivalent to $\ket{H_\pm}$.

\item[Case 5:] One can check that it is impossible for 5 or more entries of the discrete Wigner function to be negative. We have thus exhausted all possibilities. 
\end{itemize}

Let us summarize the above discussion. we considered the discrete Wigner function of an arbitrary pure qutrit state. We find that the discrete Wigner function can have at most four negative entries. By maximizing the sum-negativity in each of these cases, we are able to recover the magic states identified in the previous subsections. Of these, the states with the maximal mana are the strange state and the 
Norell state.

It would be interesting to repeat this analysis for ququints, but because a pure ququing state depends on 8 real variables, and the ququint discrete Wigner function has 25 entries, this is not feasible.
\bibliographystyle{ssg}
\bibliography{qudit-magic-states}

\end{document}